\documentclass[a4paper,USenglish,cleveref,thm-restate]{lipics-v2019}
\pdfoutput=1

\crefname{claim}{Claim}{Claims}

\usepackage{amssymb, latexsym, amsfonts, stmaryrd}
\usepackage{calc, subcaption}

\usepackage[inference]{semantic}

\usepackage{tikz}
\usetikzlibrary{arrows, calc, shapes, snakes, automata, fit, positioning, intersections}
\tikzset{>=stealth'} %

\tikzstyle{every picture} = [style=semithick]
\tikzstyle{every node}    = [font=\small]
\tikzstyle{every state}   = [thick, minimum size=1mm, inner sep=2pt]

\tikzstyle{initial}   = [initial   by arrow, initial   text=, initial   distance=4mm]
\tikzstyle{accepting} = [accepting by arrow, accepting text=, accepting distance=4mm]

\newcommand{\setN}{\mathbb{N}}
\newcommand{\setZ}{\mathbb{Z}}
\newcommand{\interval}[2]{[#1 , #2]}

\renewcommand{\vec}[1]{{\mathbf #1}}

\newcommand{\vass}{\mathcal{V}}
\newcommand{\ztest}{\mathtt{tst}}
\newcommand{\lps}{L}
\newcommand{\vloop}{\updownarrow}
\newcommand{\pvloop}{\uparrow}
\newcommand{\nvloop}{\downarrow}
\newcommand{\woca}{\mathcal{A}}
\newcommand{\weight}{\lambda}

\newcommand{\wrightarrow}[2]{%
  \xrightarrow[{\raisebox{1.25ex-\heightof{$\scriptstyle#2$}}[0pt]{$\scriptstyle#2$}}]{#1}%
}

\newcommand{\var}[1]{\mathtt{#1}}

\newcommand{\eqdef}{\stackrel{\scriptscriptstyle \mathrm{def}}{=}}

\newcommand{\rcomp}{\fatsemi}

\newcommand{\norm}[1]{\|#1\|}

\newcommand{\disp}[1]{\operatorname{disp}(#1)}

\title{%
  Reachability in Two-Dimensional Vector Addition Systems with States: One Test is for Free
}
\titlerunning{%
  Reachability in 2-dim VASS: One Test is for Free
}

\author{Jérôme Leroux}{LaBRI, Univ. Bordeaux, CNRS, Bordeaux-INP, Talence, France}{jerome.leroux@labri.fr}{}{}
\author{Grégoire Sutre}{LaBRI, Univ. Bordeaux, CNRS, Bordeaux-INP, Talence, France}{gregoire.sutre@labri.fr}{}{}

\authorrunning{J. Leroux and G. Sutre}

\Copyright{Jérôme Leroux and Grégoire Sutre}



\funding{%
  This work was supported by the grant ANR-17-CE40-0028 of the French National
  Research Agency ANR (project BRAVAS).
}

\nolinenumbers
\hideLIPIcs

\begin{document}

\maketitle

\begin{abstract}
  Vector addition system with states is an ubiquitous model of
  computation with extensive applications in
  computer science.
  The
  reachability problem for vector addition systems is central since
  many other problems reduce to that question.
  The problem is decidable and it was recently proved that the dimension of the vector addition
  system is an important parameter of the complexity. In fixed
  dimensions larger than two, the complexity is not known (with huge
  complexity gaps).
  In dimension two, the reachability problem
  was shown to be PSPACE-complete by Blondin et al. in 2015.
  We consider an extension of this model, called 2-TVASS,
  where the first counter can be tested
  for zero. This model naturally extends the classical model of one
  counter automata (OCA).
   We show that reachability is still solvable in polynomial space for 2-TVASS.
  As in the work Blondin et al.,
  our approach relies on the existence of small reachability certificates
  obtained by concatenating polynomially many cycles.
\end{abstract}

\section{Introduction}
\label{sec:introduction}
\subparagraph{Context}
Vector addition systems with states (VASS for short) is an
ubiquitous model of computation with extensive applications
in computer science. This model, equivalent to Petri nets, is defined as a finite state automaton
with transitions acting on a set of counters ranging over the
nonnegative integers by adding integers. The number of
counters is called the dimension and we write $d$-VASS for a VASS with
$d$ counters.
The central problem on VASS is the reachability problem
since many other problems are reducible to reachability
questions. This problem was first proved to be hard for the
exponential-space complexity by Lipton~\cite{lipton76} in 1976. At that
time, the decidability of the problem was open. Three years
later~\cite{HP79}, the reachability problem for $2$-VASS was proved to
be decidable
by Hopcroft and Pansiot by observing that reachability sets of
$2$-VASS are \emph{semilinear}. Dimension two
is a special case since in contrast reachability sets of
$3$-VASS are not semilinear in general.
A few years later, the
reachability problem for VASS was proved to be decidable in any
dimension by Mayr~\cite{Mayr81,Mayr84} thanks
to an algorithm simplified later by Kosaraju~\cite{Kosaraju82} and
Lambert~\cite{Lambert92}. Recently, the problem was revisited by
Leroux~\cite{Leroux10,Leroux11,Leroux12} by observing that the
reachability problem can be decided with a simple algorithm based on
semilinear inductive invariants. Despite recent improvements on
the reachability problem, the exact complexity is still
open; the known lower-bound is
Tower-hard~\cite{DBLP:conf/stoc/CzerwinskiLLLM19} and the known
upper-bound is Ackermannian-easy~\cite{DBLP:conf/lics/LerouxS19}.

\smallskip

When adding to VASS the ability to test counters for zero, the
reachability problem becomes undecidable in dimension two via a
direct simulation of two-counters Minsky
machines~\cite[Chapter~14]{minsky1967computation}.
In dimension one, the
class of VASS that we obtain by adding zero-tests are usually called one counter automata (OCA for
short). The reachability problem for that class was proved to be
NP-complete in~\cite{DBLP:conf/concur/HaaseKOW09}. The class of OCA
can be naturally extended by introducing the class of $d$-TVASS (or
just TVASS when the dimension $d$ is not fixed)
corresponding to a $d$-VASS extended with zero-tests on
the first counter. In that context, a OCA is just a $1$-TVASS.
The
reachability problem is known to be decidable for TVASS in any
dimension~\cite{DBLP:journals/entcs/Reinhardt08,DBLP:conf/mfcs/Bonnet11}, but
the complexity is open, even in dimension two.

\smallskip

In dimension two, the reachability problem for VASS is known to be
PSPACE-complete. This result was obtained thanks to a series of
results from several authors. The PSPACE lower-bound was proved
in~\cite{DBLP:journals/iandc/FearnleyJ15} and PSPACE membership
was obtained as follows (notice that the problem was recently
revisited in~\cite{DBLP:conf/mfcs/CzerwinskiLLP19}).
First of all, the reachability
relation was proved to be semilinear
in~\cite{DBLP:conf/concur/LerouxS04} by observing that it is
\emph{flattenable}, meaning that the
reachability relation can be captured by a finite set of regular
expressions, so called \emph{linear path schemes}, of the form
$\alpha_0\beta_1^*\alpha_1\cdots \beta_k^*\alpha_k$ where
$\alpha_0,\ldots,\alpha_k$ are paths and $\beta_1,\ldots,\beta_k$ are
cycles in the underlying graph of the
VASS. 
It was then proved in~\cite{BFGHM-lics15} that these regular
expressions can be exponentially bounded, and $k$ is bounded by a polynomial
in the number of states. By introducing a system of inequalities over some variables $n_1,\ldots,n_k$
counting the number of times the cycles $\beta_1,\ldots,\beta_k$ are
iterated, an exponential bound on small paths witnessing
reachability was derived from a small solution
theorem~\cite{pottier91,DBLP:journals/siammax/BoroshT92}. From such a
bound, it follows that the reachability problem is decidable in
PSPACE.

\subparagraph{Our contribution}
In this paper, we are interested in the complexity of the reachability
problem for $2$-TVASS. We successfully follow the approach used for $2$-VASS
and outlined above.
This approach is not easily lifted to $2$-TVASS,
because the presence of zero-tests breaks a fundamental property of VASS,
namely \emph{monotonicity}.
By means of new proof techniques to deal with zero-tests on a single counter,
we obtain the following results:
\begin{itemize}
\item  We show that the reachability relation of a $2$-TVASS is
flattenable. Our proof does not provide by itself any complexity bound
but it is direct and simple, and it provides a description of the reachability relation by
\emph{linear path schemes} $\alpha_0\beta_1^*\alpha_1\cdots \beta_k^*\alpha_k$.
\item We prove that these linear path schemes can be exponentially
bounded, and the number $k$ can be polynomially bounded in the number
of states of the $2$-TVASS. This bound is obtained via a detour through the
class of \emph{weighed one counter automata} (WOCA for short).
We believe that our results on WOCA may be of independent interest.
\item  We derive an exponential bound on paths witnessing
reachability thanks to a small solution theorem. From that bound, we
deduce that the reachability problem for $2$-TVASS is decidable in
polynomial space, and so is PSPACE-complete.
This is, to our knowledge, one of the few problems on extended VASS
whose precise complexity is known.
\end{itemize}

\subparagraph{Related work}
TVASS are
naturally related to other classical extensions of VASS by observing
that a reset is a ``weak test'', and a testable counter is a ``weak stack''.

By extending $d$-VASS with resets on the two first
counters, we obtain the class of $d$-RRVASS. It is known that
the reachability problem for this class is undecidable if $d\geq 3$ while it is
decidable for $d=2$~\cite{dufourd98}. Since a reset can be simulated
by a test, the class of $2$-TRVASS obtained from $2$-VASS by allowing tests on the first
counter and resets on the second
one, contains the $2$-RRVASS. In~\cite{DBLP:conf/fsttcs/FinkelLS18}, we proved that the
reachability problem for $2$-TRVASS is decidable by proving that the
reachability relation is effectively semilinear. It worth noticing
that this relation is \emph{not} flattenable, and the complexity of the
reachability problem for $2$-RRVASS and $2$-TRVASS is still open.
When dealing with the
\emph{lossy semantics} (i.e., when counters can be decreased
arbitrarily at any step of the execution), tests and resets have
exactly the same behavior. In that case, the reachability problem for
lossy Minsky machines of arbitrary dimension becomes decidable and the exact complexity is
Ackermannian complete~\cite{DBLP:conf/mfcs/Schnoebelen10}.

The class of TVASS is also related to the class of pushdown VASS (PVASS for
short) obtained by extending VASS with a stack over a
finite alphabet. A PVASS can easily simulate any TVASS since a
testable counter can be simulated with a stack. The decidability of the reachability
problem is open even for $1$-PVASS. We
proved in~\cite{DBLP:conf/icalp/LerouxST15} that the control-state reachability problem for $1$-PVASS is
decidable. The complexity is still open.

\smallskip

Due to space constraints,
some proofs are missing and some proofs are only sketched.
Detailed proofs can be found in appendix.

\section{Flattenability of 2-TVASS}
\label{sec:flattenability-of-2-tvass}
\subparagraph{Preliminaries}

The usual sets of \emph{integers} and \emph{nonnegative integers} are denoted by $\setZ$ and $\setN$,
respectively.
For any $a, b \in \setZ$, we let $\interval{a}{b} \eqdef \{z \in \setZ \mid a \leq z \leq b\}$.
A $d$-dimensional \emph{vector} of integers is a tuple $\vec{v} = (v_1, \ldots, v_d)$ in $\setZ^d$.
Its $i$th \emph{component} $v_i$ is also written $\vec{v}(i)$.
We denote by $\norm{\vec{v}}$ its \emph{infinity norm} $\max \{|v_1|, \ldots, |v_d|\}$.
A \emph{word} over some alphabet $\Sigma$ is a finite sequence $w = a_1 \cdots a_n$ of elements $a_i \in \Sigma$.
The \emph{length} of $w$ is $|w| \eqdef n$.
Given two binary relations $R$ and $S$ over some set,
we let $R \rcomp S \eqdef \{(x, z) \mid \exists y : x \,R\, y \,S\, z\}$
denote their \emph{relational composition}.
The \emph{powers} of a binary relation $R$ are inductively defined by
$R^1 \eqdef R$ and $R^{n+1} \eqdef R \rcomp R^n$.

\subparagraph{Vector Addition Systems with States and One Test}

A TVASS is a vector addition system with states (VASS) such that
the first counter can be tested for zero.
Formally, a \emph{$d$-dimensional TVASS} (shortly called a \emph{$d$-TVASS}),
is a triple $\vass = (Q, \Sigma, \Delta)$ where
$Q$ is a finite nonempty set of \emph{states},
$\Sigma \subseteq \setZ^d \cup \{\ztest\}$ is a finite set of \emph{actions},
and $\Delta \subseteq Q \times \Sigma \times Q$ is a finite set of \emph{transitions}.
Even though they are not mentioned explicitly,
$\vass$ implicitly comes with $d$ counters $\var{c}_1, \ldots, \var{c}_d$
whose values range over nonnegative integers.
Actions in $\Sigma$ are
either \emph{addition} actions $\vec{a} \in \setZ^d$ or
the \emph{zero-test} action $\ztest$.
Intuitively,
an addition action $\vec{a} = (a_1, \ldots, a_d)$ performs the instruction
$(\var{c}_1, \ldots, \var{c}_d) \leftarrow (\var{c}_1 + a_1, \ldots, \var{c}_d + a_d)$,
provided that all counters remain nonnegative ;
the zero-test action $\ztest$ tests the first counter for zero and leaves all counters unchanged.
We let
$A \eqdef \{(p, \sigma, q) \in \Delta \mid \sigma \in \setZ^d\}$
and
$T \eqdef \{(p, \sigma, q) \in \Delta \mid \sigma = \ztest\}$
denote the sets of \emph{addition} transitions and \emph{zero-test} transitions,
respectively.
The notation $\norm{\Sigma}$ stands for $\max_{\vec{a}} \norm{\vec{a}}$ where
$\vec{a}$ ranges over addition actions (or $\{\vec{0}\}$ if there are none).
A \emph{$d$-dimensional VASS} (shortly called a \emph{$d$-VASS}) is a $d$-TVASS
whose set of actions $\Sigma$ excludes $\ztest$, i.e., $\Sigma \subseteq \setZ^d$.

\smallskip

We define the operational semantics of a $d$-TVASS
$\vass = (Q, \Sigma, \Delta)$ as follows.
A \emph{configuration} of $\vass$ is a pair $(q, \vec{x})$
where $q \in Q$ is a state and
$\vec{x} \in \setN^d$ is a vector
denoting the contents of the counters $\var{c}_1, \ldots, \var{c}_d$.
For the sake of readability,
configurations $(q, \vec{x})$ are written $q(\vec{x})$ in the sequel.
For each transition $\delta \in \Delta$,
we let $\xrightarrow{\delta}$ denote the least binary relation
over configurations satisfying the following rules:
$$
\begin{array}{@{}c@{\qquad\quad}c@{}}
  \inference
    {\delta = (p, \vec{a}, q) \qquad \vec{x} \in \setN^d \qquad \vec{x} + \vec{a} \geq \vec{0}}
    {p(\vec{x}) \xrightarrow{\delta} q(\vec{x} + \vec{a})}
  &
  \inference
    {\delta = (p, \ztest, q) \qquad \vec{x} \in \setN^d \qquad \vec{x}(1) = 0}
    {p(\vec{x}) \xrightarrow{\delta} q(\vec{x})}
\end{array}
$$

Given a word $\pi = \delta_1 \cdots \delta_n$ of transitions $\delta_i \in \Delta$,
we denote by $\xrightarrow{\pi}$ the binary relation over configurations
defined as the relational composition
$\xrightarrow{\delta_1} \rcomp \cdots \rcomp \xrightarrow{\delta_n}$.
The relation $\xrightarrow{\varepsilon}$ denotes the identity relation on configurations.
Given a subset $L \subseteq \Delta^*$,
we let $\xrightarrow{L}$ denote the union $\bigcup_{\pi \in L} \xrightarrow{\pi}$.
The relation $\xrightarrow{\Delta^*}$, also written $\xrightarrow{*}$,
is called the \emph{reachability relation} of $\vass$.
Observe that $\xrightarrow{*}$ is the reflexive-transitive closure of
the \emph{step} relation ${\rightarrow} \eqdef {\xrightarrow{\Delta}}$.

\smallskip

A \emph{run} is a finite, alternating sequence
$(q_0(\vec{x}_0), \delta_1, q_1(\vec{x}_1), \ldots, \delta_n, q_n(\vec{x}_n))$
of configurations and transitions,
satisfying $q_{i-1}(\vec{x}_{i-1}) \xrightarrow{\delta_i} q_i(\vec{x}_i)$
for all $i \in \interval{1}{n}$.
Note that this condition entails that
$q_0(\vec{x}_0) \xrightarrow{\delta_1 \cdots \delta_n} q_n(\vec{x}_n)$.
The word $\delta_1 \cdots \delta_n$ is called the \emph{trace} of the run
and $n$ is its \emph{length}.

\begin{figure}[t]
  \centering
  \begin{tikzpicture}[node distance=2.5cm, text centered, ->, bend angle=30]
    \node[state] (A)                           {$A$};
    \node[state] (B) [right of=A, xshift=10mm] {$B$};

    \draw[->]
      (A) edge [bend left] node [above] {$\var{c}_1 == 0$} (B)
      (B) edge [bend left] node [below] {$\var{c}_1 \leftarrow \var{c}_1 + 1$} (A)
    ;

    \draw[->]
      (A) edge [loop left, looseness=10, out=-150, in=150] node [left]
      {$(\var{c}_1, \var{c}_2) \leftarrow (\var{c}_1 - 3, \var{c}_2 + 4)$} (A)
    ;

    \draw[->]
      (B) edge [loop right, looseness=10, out=30, in=-30] node [right]
      {$(\var{c}_1, \var{c}_2) \leftarrow (\var{c}_1 + 1, \var{c}_2 - 1)$} (B)
    ;
  \end{tikzpicture}
  \caption{%
    A simple $2$-dimensional TVASS with actions written in verbose pseudo-code
    (to help the reader).
    Instructions of the form
    $(\var{c}_1, \var{c}_2) \leftarrow (\var{c}_1 + a_1, \var{c}_2 + a_2)$
    stand for addition actions $(a_1, a_2)$.
    The instruction $\var{c}_1 == 0$ stands for the zero-test action $\ztest$.
  }
  \label{fig:2TVASS-example}
\end{figure}
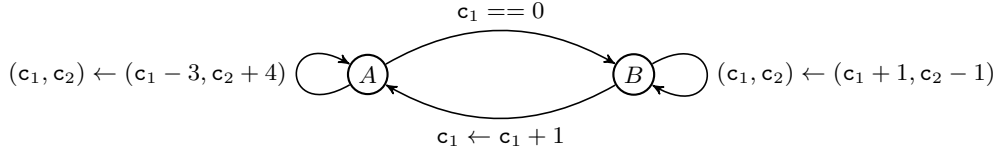

\begin{example}
  \label{ex:2TVASS}
  Consider the $2$-TVASS depicted in Figure~\ref{fig:2TVASS-example}.
  There are two states, namely $A$ and $B$, and four transitions, namely:
  $$
  \begin{array}[b]{@{}rcl@{\,}c@{\,}c@{\,}c@{\,}r@{\qquad\qquad}rcl@{\,}c@{\,}c@{\,}c@{\,}r@{}}
    \delta_{AA} & = & (A &,& (-3, 4) &,& A) & \delta_{BB} & = & (B &,& (1, -1) &,& B)\\
    \delta_{AB} & = & (A &,& \ztest &,&  B) & \delta_{BA} & = & (B &,& (1, 0)  &,& A)
  \end{array}\:.
  $$
  Starting from the configuration $A(3, 5)$,
  the zero-test transition $\delta_{AB}$ cannot be taken as the first counter
  is not zero.
  But we can take the loop on $A$ and reach the configuration $A(0, 9)$,
  which is formally written as the step $A(3, 5) \xrightarrow{\delta_{AA}} A(0, 9)$.
  We may then move to $B$ via the zero-test transition,
  take the loop on $B$ four times, and get back to $A$.
  This yields the run
  $\rho = (A(3, 5), \delta_{AA}, A(0, 9), \delta_{AB}, B(0, 9), \delta_{BB}, B(1, 8), \ldots, \delta_{BB}, B(4, 5), \delta_{BA}, A(5, 5))$.
  The trace of this run is $\pi = \delta_{AA} \delta_{AB} (\delta_{BB})^4  \delta_{BA}$,
  and so we have $A(3, 5) \xrightarrow{\pi} A(5, 5)$.

  \smallskip

  The run $\rho$ witnesses the fact that $A(3, 5) \xrightarrow{*} A(5, 5)$.
  In a standard $2$-VASS, i.e., without zero-test,
  the run $\rho$ could be ``replayed'' from the larger configuration $A(5, 5)$.
  More precisely, we would have $A(3, 5) \xrightarrow{\pi} A(5, 5) \xrightarrow{\pi} A(7, 5)$.
  This is not the case in our $2$-TVASS.
  Even though $A(3, 5) \xrightarrow{\pi} A(5, 5)$,
  it does not hold that $A(5, 5) \xrightarrow{\pi} A(7, 5)$.
  Indeed,
  $A(5, 5) \xrightarrow{\delta_{AA}} A(2, 9)$ and the
  zero-test transition $\delta_{AB}$ cannot be taken from $A(2, 9)$.

  \smallskip

  We cannot replay $\rho$ from the larger configuration $A(5, 5)$.
  Nonetheless,
  the configuration $A(7, 5)$ is reachable from $A(3, 5)$ in our $2$-TVASS.
  In fact,
  it holds that $A(3, 5) \xrightarrow{*} A(3 + 2k, 5)$ for every $k \in \setN$.
  This property will be shown in \cref{ex:LPS}.
  \qed
\end{example}

\subparagraph{Linear Path Schemes and Flattenability}

Consider a $d$-TVASS $\vass = (Q, \Sigma, \Delta)$.
A \emph{path} from a state $p \in Q$ to a state $q \in Q$ is
either the empty word $\varepsilon$
or a nonempty word $\delta_1 \cdots \delta_n$ of transitions,
with $\delta_i = (p_i, \sigma_i, q_i)$,
such that $p_0 = p$, $q_n = q$ and $q_{i-1} = p_i$
for all $i \in \interval{1}{n}$.
Note that for every word $\pi \in \Delta^*$,
if the relation $\xrightarrow{\pi}$ is not empty then $\pi$ is a path.
The converse does not hold in general
(but it holds if $\pi \in A^*$, i.e., if no zero-test occurs in $\pi$).
A \emph{cycle} on a state $q \in Q$ is a path from $q$ to $q$.

\smallskip

A \emph{linear path scheme}
from a state $p \in Q$ to a state $q \in Q$
is a regular expression $\lps$ of the form
$
\lps \ = \ \alpha_0 \beta_1^* \alpha_1 \cdots \beta_k^* \alpha_k
$
where $\alpha_0 \beta_1 \alpha_1 \cdots \beta_k \alpha_k$ is a path from $p$ to $q$
and each $\beta_i$ is a cycle.
We call $\beta_1, \ldots, \beta_k$ the \emph{cycles} of $\lps$.
Its \emph{length} is
$|\lps| \eqdef |\alpha_0 \beta_1 \alpha_1 \cdots \beta_k \alpha_k|$ and
its \emph{$*$-length} is $|\lps|_* \eqdef k$.
We slightly abuse notation and also write $\lps$ for
the language associated to a linear path scheme $\lps$.

\begin{example}
  \label{ex:LPS}
  We claimed at the end of Example~\ref{ex:2TVASS} that
  $A(3, 5) \xrightarrow{*} A(3 + 2k, 5)$ for every $k \in \setN$.
  The case where $k = 0$ is trivial,
  so let us prove this property assuming that $k > 0$.
  Consider the linear path scheme
  $$
  \lps \ = \
  \delta_{AA}
  \cdot
  (\delta_{AB} \delta_{BB} \delta_{BB} \delta_{BA} \delta_{AA})^*
  \cdot
  \delta_{AB}
  \cdot
  (\delta_{BB})^*
  \cdot
  \delta_{BA}
  \:.
  $$
  Note in passing that $\lps$ has length $|\lps| = 9$ and $*$-length $|\lps|_* = 2$.
  To simplify notation,
  let $\pi \eqdef \delta_{AB} \delta_{BB} \delta_{BB} \delta_{BA} \delta_{AA}$.
  Observe that $A(0, x) \xrightarrow{\pi} A(0, x + 2)$
  for every $x \geq 2$.
  It follows that
  $
  A(3, 5)
  \xrightarrow{\delta_{AA}}
  A(0, 9)
  \xrightarrow{\pi^{k-1}}
  A(0, 2k + 7)
  \xrightarrow{\delta_{AB}}
  B(0, 2k + 7)
  \xrightarrow{(\delta_{BB})^{2k + 2}}
  B(2k + 2, 5)
  \xrightarrow{\delta_{BA}}
  A(2k + 3, 5)
  $.
  We have shown that $A(3, 5) \xrightarrow{\lps} A(3 + 2k, 5)$ for every $k > 0$.
  \qed
\end{example}

A binary relation $R$ over configurations is called
\emph{flattenable}\footnote{%
  The same notion is often called \emph{flattable} in the literature.
  It was simply called \emph{flat} in~\cite{DBLP:conf/concur/LerouxS04}.
} if there exists a finite set $\Lambda$ of linear path schemes such that
$R \subseteq \bigcup_{L \in \Lambda} {\xrightarrow{L}}$.
It is readily seen that the class of flattenable binary relations is
closed under union and relational composition.
We say that a $d$-TVASS $\vass$ is \emph{flattenable} when
its reachability relation $\xrightarrow{*}$ is flattenable.

\subparagraph{Flattenability of TVASS in Dimension Two}

We showed sixteen years ago in~\cite{DBLP:conf/concur/LerouxS04} that
$2$-VASS are flattenable.
Our approach was refined ten years later by Blondin et al. to provide
bounds on the resulting linear path schemes and
to show that the reachability problem for $2$-VASS is solvable in
polynomial space~\cite{BFGHM-lics15}.

\begin{theorem}[\cite{DBLP:conf/concur/LerouxS04,BFGHM-lics15}]
  \label{thm:2-vass-flat}
  Every $2$-VASS is flattenable.
  Furthermore,
  for every configurations $p(\vec{x})$ and $q(\vec{y})$
  of a $2$-VASS $\vass = (Q, \Sigma, \Delta)$ such that
  $p(\vec{x}) \xrightarrow{*} q(\vec{y})$,
  there exists a linear path scheme $L$ with
  $|L| \leq (|Q| + \norm{\Sigma})^{O(1)}$ and
  $|L|_* \leq O(|Q|^2)$
  such that $p(\vec{x}) \xrightarrow{L} q(\vec{y})$.
\end{theorem}

The remainder of this section is devoted to the extension of
the first part of \cref{thm:2-vass-flat} to $2$-TVASS.
The existence of small linear path schemes witnessing flattenability
will be shown in
\cref{sec:weighted-one-counter-automata,sec:succint-flattenability-of-2-tvass},
and ensuing complexity results will be presented in
\cref{sec:linear-path-scheme-to-equations,sec:complexity-results}.

\smallskip

Consider a $2$-TVASS $\vass = (Q, \Sigma, \Delta)$.
We introduce, for each state $q \in Q$,
the binary relation $\vloop_q$ over configurations defined by
$
{\vloop_q} = \{(q(0, x), q(0, y)) \mid q(0, x) \xrightarrow{*} q(0, y)\}
$.
This relation is called the \emph{vertical loop relation} on $q$.
We let $\vloop$ denote the union $\bigcup_{q \in Q} {\vloop_q}$.
We first provide a decomposition of $\xrightarrow{*}$ in terms of
$\xrightarrow{A^*}$, $\xrightarrow{T}$ and $\vloop$.
Recall that $A$ and $T$ are the sets of addition transitions and
zero-test transitions, respectively.

\begin{restatable}{lemma}{vloopflatnesstoflatness}
  \label{lem:vloop-flatness-to-flatness}
  It holds that
  $
  {\xrightarrow{*}} \subseteq
  \left( {\xrightarrow{A^*}} \cup {\xrightarrow{T}} \cup {\vloop} \right)^{2|Q| + 1}
  $.
\end{restatable}

The binary relations $\xrightarrow{A^*}$ and $\xrightarrow{T}$ are already
known to be flattenable.
This is a consequence of \cref{thm:2-vass-flat} for the former,
and flattenability is obvious for the latter.
As flattenable binary relations are closed under union and relational composition,
it remains to show that $\vloop$ is flattenable.
Vertical loops $q(0, x) \xrightarrow{*} q(0, y)$
either increase the second counter (i.e., $y > x$),
or decrease it (i.e., $y < x$),
or leave it unchanged (i.e., $y = x$).
Let $\pvloop_q$ and $\nvloop_q$ denote the subrelations of
$\vloop_q$ that correspond to the first and second cases, respectively.
We first prove that the relations $\pvloop_q$ are flattenable.

\smallskip

Fix a state $q \in Q$ and assume that $\pvloop_q$ is not empty
(otherwise it is trivially flattenable).
We introduce the sequence $(D_x)_{x \in \setN}$ of subsets of $\setN$ defined by
$$
D_x \ = \ \{d \in \setN \mid q(0, x) \xrightarrow{*} q(0, x + d)\}
$$
We derive from the monotonicity of $2$-TVASS with respect to the second counter
that $D_0 \subseteq D_1 \subseteq D_2 \cdots$ and that\footnote{%
  The sum $A + B$ of two subsets $A, B \subseteq \setZ$ is defined as
  $\{a + b \mid a \in A \wedge b \in B\}$.
}
$(D_x + D_x) \subseteq D_x$ for every $x \in \setN$.
These two properties entail that the sequence $(D_x)_{x \in \setN}$ is
ultimately stationary.
Indeed, suppose by contradiction that there is an infinite increasing
subsequence $\{0\} \subsetneq D_{x_0} \subsetneq D_{x_1} \subsetneq D_{x_2} \subsetneq \cdots$.
We may extract $c, d_1, d_2, \ldots$ such that
$c \in D_{x_0}$, $c > 0$, and
$d_i \in D_{x_i} \setminus D_{x_{i-1}}$ for all $i > 0$.
By the pigeonhole principle,
some congruence class modulo $c$ contains infinitely many $d_i$.
So there exists $0 < i < j$ and $k \in \setN$ such that $d_j = d_i + k c$.
As $d_i$ and $c$ are both in $D_{x_i}$,
we get from $(D_{x_i} + D_{x_i}) \subseteq D_{x_i}$ that $d_j \in D_{x_i}$,
which is impossible since $D_{x_i} \subseteq D_{x_{j-1}}$ and $d_j \not\in D_{x_{j-1}}$.
We have shown that there exists $t \in \setN$ such that
$D_x = D_t$ for all $x \geq t$.

\smallskip

Recall that $\pvloop_q$ was assumed to be nonempty.
So there exists $h \geq 0$, $m > 0$ and a run from $q(0, h)$ to $q(0, h + m)$.
Let $\beta$ denote the trace of this run.
Note that $\beta$ is a nonempty cycle on $q$.
We derive from $q(0, h) \xrightarrow{\beta} q(0, h + m)$ that
$(d + m) \in D_x$ for all $x \in \setN$ and $d \in D_x$ with $d \geq h$.
It follows that each $D_x$ may be decomposed into $D_x = F_x \cup (B_x + \setN m)$
where $F_x$ and $B_x$ are finite subsets of $\setN$ such that $b \geq h$ for all $b \in B_x$.
For every $d \in D_x$,
let $\alpha_{x, d}$ denote the trace of some run from $q(0, x)$ to $q(0, x + d)$.
Consider the finite set $\Lambda$ of linear path schemes defined by
$$
\Lambda \ = \ \bigcup_{x \leq t} \Lambda_x
\qquad\qquad
\text{and}
\qquad\qquad
\Lambda_x \ = \ \{\alpha_{x, f} \mid f \in F_x\} \ \cup \ \{\alpha_{x, b} \beta^* \mid b \in B_x\}\:.
$$
Observe that $\Lambda$ is finite as it collects the linear path schemes in $\Lambda_x$
only for $x \leq t$.
The linear path schemes in $\Lambda_x$ with $x > t$ are redundant
because of the above-established stabilization property of $(D_x)_{x \in \setN}$.
We obtain the following lemma by construction.

\begin{restatable}{lemma}{increasingvloopflattenable}
  \label{lem:increasing-vloop-flattenable}
  It holds that ${\pvloop_q} \subseteq \bigcup_{L \in \Lambda} {\xrightarrow{L}}$,
  hence, the relation $\pvloop_q$ is flattenable.
\end{restatable}

Notice that a decreasing vertical loop $q(0, x) \xrightarrow{*} q(0, y)$ with $y < x$
is an increasing vertical loop in the $2$-TVASS $\overline{\vass}$ obtained from $\vass$ by
reversing the effect of each transition, i.e.,
$\overline{\Delta} = \{\overline{\delta} \mid \delta\in\Delta\}$ where
$\overline{(p, \vec{a}, q)} = (q, -\vec{a}, p)$ and $\overline{(p, \ztest, q)} = (q, \ztest, p)$.
Put differently,
the relation $\nvloop_q$ in $\vass$ coincides with the relation $\pvloop_q$ in $\overline{\vass}$.
By applying \cref{lem:increasing-vloop-flattenable} to $\overline{\vass}$ and
taking the mirror image of the resulting linear path schemes,
we get that the relation $\nvloop_q$ in $\vass$ is also flattenable.
Since ${\vloop} = \bigcup_{q \in Q} {\vloop_q}$ and
${\vloop_q} \subseteq {\pvloop_q} \cup {\nvloop_q} \cup {\xrightarrow{\varepsilon}}$ for every $q \in Q$,
we obtain that $\vloop$ is flattenable.
Together with \cref{lem:vloop-flatness-to-flatness} and \cref{thm:2-vass-flat},
this concludes the proof of the following theorem.

\begin{theorem}
  \label{thm:2-tvass-flat}
  Every $2$-TVASS is flattenable.
\end{theorem}

We have presented in this section a direct and simple proof
that the reachability relation of every $2$-TVASS is flattenable.
This result entails, in particular, that
the reachability relation is effectively semilinear for $2$-TVASS
(which was already known~\cite{DBLP:conf/fsttcs/FinkelLS18})
and computable by cycle acceleration techniques (see, e.g.,~\cite{DBLP:conf/atva/LerouxS05}).

\smallskip

To derive complexity results from flattenability,
we need to bound the length of linear path schemes witnessing flattenability.
This requires a finer analysis of $2$-TVASS runs than what was done for \cref{thm:2-tvass-flat}
(the latter will not be used in the remainder of the paper).

\section{A Detour via Weighted One-Counter Automata}
\label{sec:weighted-one-counter-automata}
We have given in the previous section a simple proof that $2$-TVASS are flattenable.
This proof provides no bound on the length of the resulting linear path schemes, though.
To obtain small linear path schemes witnessing flattenability,
we take a detour via weighted one-counter automata.
The rationale is that a $2$-TVASS behaves like a one-counter automaton
equipped with an additional counter (that cannot be tested for zero).
When this additional counter is allowed to become negative,
actions on it can be seen as weights.
We show in this section that,
in a weighted one-counter automaton,
the reachable weights between two mutually reachable configurations
$p(0)$ and $q(0)$ can be obtained via small linear path schemes.
This will yield,
for $2$-TVASS,
small linear path schemes for the reachability subrelations
$p(0, x) \xrightarrow{*} q(0, y)$ such that $x$ and $y$ are large and
$q(0, y) \xrightarrow{*} p(0, z)$ for some $z$.
For simplicity,
we consider weighted one-counter automata where addition actions and weights
are in $\{-1, 0, 1\}$

\smallskip

A \emph{weighted one-counter automaton} (shortly called a \emph{WOCA}),
is a quadruple $\woca = (Q, \Sigma, \Delta, \weight)$ where
$(Q, \Sigma, \Delta)$ is a $1$-TVASS such that $\Sigma = \{-1, 0, 1, \ztest\}$ and
$\weight : \Delta \rightarrow \{-1, 0, 1\}$ is a \emph{weight} function.
All notions defined in \cref{sec:flattenability-of-2-tvass} for $1$-TVASS
naturally carry over to WOCA.
The weight function is extended to words in $\Delta^*$ by
$\weight(\delta_1 \cdots \delta_n) = \weight(\delta_1) + \cdots + \weight(\delta_n)$.
The \emph{weight} of a run is the weight of its trace.
For notational convenience,
we write $p(x) \wrightarrow{\pi}{w} q(y)$ when
$p(x) \xrightarrow{\pi} q(y)$ and $w = \weight(\pi)$.
Similarly,
we let $p(x) \wrightarrow{*}{w} q(y)$ stand for the existence of
$\pi \in \Delta^*$ such that $p(x) \wrightarrow{\pi}{w} q(y)$.

\smallskip

As mentioned before, this section is devoted to the proof that,
for any states $p$ and $q$ in a WOCA,
the weights $w \in \setZ$ such that $p(0) \wrightarrow{*}{w} q(0) \xrightarrow{*} p(0)$
can be obtained via small linear path schemes (see \cref{lem:pumping-short-lps}).
We start with two pumping lemmas on runs of one-counter automata.
The first one, \cref{lem:hill-cutting},
can be seen as an iterated version of the pumping lemma
for one-counter languages due to Latteux~\cite{DBLP:journals/jcss/Latteux83}.
It is an easy consequence of the classical \emph{hill-cutting}
technique for one-counter automata
(often attributed to Valiant and Paterson~\cite{DBLP:journals/jcss/ValiantP75}).
The second one, \cref{lem:pumping-short-cycles},
tunes the hill-cutting technique so as to obtain short extracted cycles.
This second pumping lemma is crucial to obtain linear path schemes with short cycles.
The hill-cutting techniques used in both lemmas are illustrated in \cref{fig:pumping-lemmas}.

\smallskip

We assume for the remainder of this section that
$\woca = (Q, \Sigma, \Delta, \weight)$ is a WOCA and that $p, q \in Q$ are states of $\woca$.

\begin{figure}[t]
\begin{subfigure}[b]{.48\textwidth}
  \centering
  \begin{tikzpicture}[scale=0.13]
    \coordinate (r1) at (7.67,5);
    \coordinate (r2) at (13.37,12);
    \coordinate (r3) at (18.36,17.5);
    \coordinate (s3) at (21.29,17.5);
    \coordinate (s2) at (22.49,12);
    \coordinate (s1) at (38.49,5);

    \fill[lightgray!33] (r1 |- 0,0) -- (r1) -- (s1) -- (s1 |- 0,0) -- cycle;
    \fill[lightgray!67] (r2 |- 0,0) -- (r2) -- (s2) -- (s2 |- 0,0) -- cycle;
    \fill[lightgray]    (r3 |- 0,0) -- (r3) -- (s3) -- (s3 |- 0,0) -- cycle;

    \draw[thick, ->] (0,0) -- (43,0) node [below left] {\textit{time}};
    \draw[thick, ->] (0,0) -- (0,23) node [above] {\textit{counter}};

    \coordinate (q1) at (0,10);
    \coordinate (q2) at (0,20);
    \node at (q2) [left] {$m|Q|^2$};
    \draw[thin, dashed, -] (q2) -- (q2 -| 40,0);

    \draw plot [smooth] coordinates {
      (0,0)
      (2,4)
      (3,2)
      (4,4)
      (6,0.2)
      (10,12)
      (12,8)
      (15,17)
      (17,13)
      (20,22)
      (23,11) %
      (26,22)
      (29,13)
      (31,17)
      (34,8)
      (36,12)
      (40,0)
    };

    \fill (r1) circle (3mm) node [left,  xshift=.5mm]  {$r$};
    \fill (r2) circle (3mm) node [left,  xshift=.5mm]  {$r$};
    \fill (r3) circle (3mm) node [left,  xshift=.5mm]  {$r$};
    \fill (s3) circle (3mm) node [right, xshift=-.5mm] {$s$};
    \fill (s2) circle (3mm) node [right, xshift=.5mm]  {$s$};
    \fill (s1) circle (3mm) node [right, xshift=-.5mm] {$s$};
  \end{tikzpicture}
\end{subfigure}%
\hfill
\begin{subfigure}[b]{.48\textwidth}
  \centering
  \begin{tikzpicture}[scale=0.13]
    \coordinate (i0) at (12.8,0);
    \coordinate (j0) at (33.2,0);
    \node at (i0) [below] {$t_1$};
    \node at (j0) [below] {$t_2$};
    \fill[lightgray!67] (i0) -- (i0 |- 0,10) -- (j0 |- 0,10) -- (j0) -- cycle;

    \coordinate (iq2) at (19.02,0);
    \coordinate (jq2) at (26.98,0);
    \fill[lightgray!67] (i0  |- 0,20) -- (i0  |- 0,23) -- (iq2 |- 0,23) -- (iq2 |- 0,20) -- cycle;
    \fill[lightgray!67] (jq2 |- 0,20) -- (jq2 |- 0,23) -- (j0  |- 0,23) -- (j0  |- 0,20) -- cycle;

    \draw[thick, ->] (0,0) -- (43,0) node [below left] {\textit{time}};
    \draw[thick, ->] (0,0) -- (0,23) node [above] {\textit{counter}};

    \coordinate (q1) at (0,10);
    \coordinate (q2) at (0,20);
    \node at (q1) [left] {$|Q|^2$};
    \node at (q2) [left] {$2|Q|^2$};
    \draw[thin, dashed, -] (q1) -- (q1 -| 40,0);
    \draw[thin, dashed, -] (q2) -- (q2 -| 40,0);

    \draw plot [smooth] coordinates {
      (0,0)
      (2,4)
      (3,2)
      (4,4)
      (6,0.2)
      (10,12)
      (12,8)
      (15,17)
      (17,13)
      (20,22)
      (23,11) %
      (26,22)
      (29,13)
      (31,17)
      (34,8)
      (36,12)
      (40,0)
    };

    \coordinate (r1) at (13.92,14);
    \coordinate (r2) at (18.36,17.5);
    \coordinate (s2) at (27.64,17.5);
    \coordinate (s1) at (32.08,14);

    \fill (r1) circle (3mm) node [left,  xshift=.5mm]  {$r$};
    \fill (r2) circle (3mm) node [left,  xshift=.5mm]  {$r$};
    \fill (s2) circle (3mm) node [right, xshift=-.5mm] {$s$};
    \fill (s1) circle (3mm) node [right, xshift=-.5mm] {$s$};
  \end{tikzpicture}
\end{subfigure}%
  \caption{%
    Illustration of the pumping lemmas for one-counter automata
    (\cref{lem:hill-cutting} on the left and \cref{lem:pumping-short-cycles} on the right).
    The curves show the evolution of the counter along a run.
    Gray areas denote forbidden zones for the counter.
  }
  \label{fig:pumping-lemmas}
\end{figure}
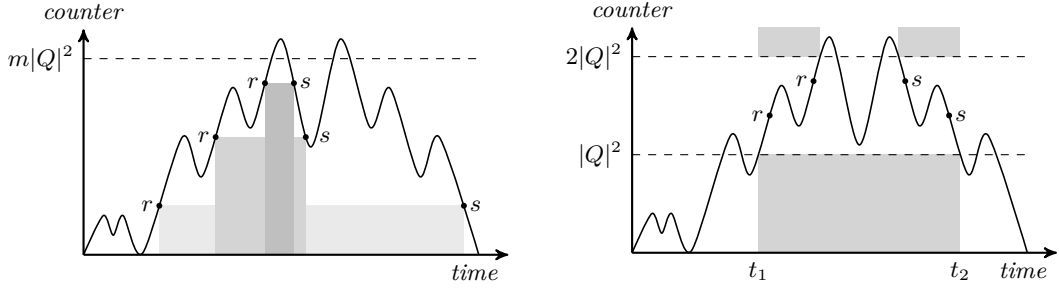

\begin{restatable}{lemma}{hillcutting}
  \label{lem:hill-cutting}
  If $p(0) \xrightarrow{\pi} q(0)$ then
  for every $m > 0$ such that $|\pi| \geq m^2 |Q|^3$,
  there exists
  a factorization $\pi = \alpha \beta_1 \cdots \beta_m \gamma \theta_m \cdots \theta_1 \eta$,
  with $\beta_i \theta_i \neq \varepsilon$ for all $i \in \interval{1}{m}$,
  verifying
  $$
  p(0)
  \xrightarrow{\alpha \beta_1^{n_1} \cdots \beta_m^{n_m} \gamma \theta_m^{n_m} \cdots \theta_1^{n_1} \eta}
  q(0)
  $$
  for every $n_1, \ldots, n_m \in \setN$.
\end{restatable}
\begin{proof}[{\proofname} Sketch]
  If the counter remains below $m|Q|^2$ then
  some configuration repeats at least $m+1$ times,
  and the subruns in-between can be iterated arbitrarily many times.
  The cycles $\beta_i$ come from these subruns and the cycles $\theta_i$ are empty.
  Otherwise,
  the run contains a ``high hill'' and we extract,
  for each counter value in $\interval{0}{m|Q|^2}$,
  a pair of configurations with this counter value
  (see \cref{fig:pumping-lemmas} (left)).
  This extraction proceeds
  from the inside of the hill towards the outside.
  Some pair of states $(r, s)$ necessarily occurs $m+1$ times in this extraction.
  The subruns between the $r$ configurations provide the cycles $\beta_i$ and
  the subruns between the $s$ configurations provide the cycles $\theta_i$.
  The extraction guarantees that these cycles can be iterated arbitrarily many times.
\end{proof}

\begin{restatable}{corollary}{shortocaruns}
  \label{cor:short-oca-runs}
  If $p(x) \xrightarrow{*} q(y)$ then
  $p(x) \xrightarrow{\pi} q(y)$ for some $\pi \in \Delta^*$ such that $|\pi| < (|Q| + x + y)^3$.
\end{restatable}

\begin{restatable}{lemma}{pumpingshortcycles}
  \label{lem:pumping-short-cycles}
  If $p(0) \xrightarrow{\pi} q(0)$ with $|\pi| \geq 2|Q|^3$
  then there exists
  $r, s \in Q$, $x, d \in \setN$,
  and a factorization $\pi = \alpha \beta \gamma \theta \eta$,
  with $\beta \theta \neq \varepsilon$ and no zero-test transition in $\gamma$,
  such that
  $x+d \leq 2|Q|^2$, $|\beta \theta| \leq 2|Q|^3$ and verifying
  $$
  p(0)
  \xrightarrow{\alpha}
  r(x)
  \xrightarrow{\beta^n}
  r(x + n d)
  \xrightarrow{\gamma}
  s(x + n d)
  \xrightarrow{\theta^n}
  s(x)
  \xrightarrow{\eta}
  q(0)
  $$
  for every $n \in \setN$.
\end{restatable}
\begin{proof}[{\proofname} Sketch]
  If the counter remains below $2|Q|^2$ then
  some configuration repeats at least twice.
  So there is a subrun of length at most $2|Q|^3$
  from some configuration $r(x)$ to the same configuration $r(x)$,
  and this subrun can be iterated arbitrarily many times.
  The cycle $\beta$ comes from this subrun and the cycle $\theta$ is empty.
  Otherwise,
  the run contains a ``high hill'' and we extract,
  for each counter value in $\interval{|Q|^2}{2|Q|^2}$,
  a pair of configurations with this counter value
  (see \cref{fig:pumping-lemmas} (right)).
  For the counter value $|Q|^2$,
  the pair of configurations is extracted
  from the inside of the hill towards the outside.
  Let $t_1$ and $t_2$ denote the positions of these configurations.
  For the counter values in $\interval{|Q|^2+1}{2|Q|^2}$,
  this extraction proceeds
  from the outside of the hill
  --- but limited to $\interval{t_1}{t_2}$ ---
  towards the inside.
  Some pair of states $(r, s)$ necessarily occurs twice in this extraction.
  The subrun between the two $r$ configurations provides the cycle $\beta$ and
  the subrun between the two $s$ configurations provides the cycle $\theta$.
  The extraction guarantees that these cycles can be iterated arbitrarily many times.
  By construction,
  the counter remains below $2|Q|^2$ in the subrun providing the cycle $\beta$
  (except possibly for the last configuration).
  If $\beta > |Q|^3$ then some configuration repeats at least twice
  in this subrun and we can proceed as in the first case of the proof.
  The same reasoning can also be applied to the cycle $\theta$.
\end{proof}

We now exploit the two previous pumping lemmas to obtain short runs with appropriate weights.
First,
we show in \cref{lem:short-runs-with-positive-weight} that
if there is a run from $p(0)$ to $q(0)$ with
positive (resp. negative) weight, then there is a short one.
In fact, this lemma will be used in the particular case where $p = q$
to get short cyclic runs with positive (resp. negative) weights.
Second,
we show in \cref{lem:short-runs-with-weight-modulo} that,
assuming that $p(0)$ and $q(0)$ are mutually reachable,
if there is a run from $p(0)$ to $q(0)$ whose
weight is in a given congruence class modulo $m > 0$,
then there is a short one and
the weight difference between the two runs can be
``qualitatively compensated'' by a cyclic run on $q(0)$.

\begin{lemma}
  \label{lem:short-runs-with-positive-weight}
  If $p(0) \wrightarrow{*}{w} q(0)$ for some $w \neq 0$ then
  $p(0) \wrightarrow{\pi}{u} q(0)$ for some $u \neq 0$ having the same sign as $w$
  and some $\pi \in \Delta^*$ such that $|\pi| \leq 539|Q|^9$.
\end{lemma}
\begin{proof}
  We only consider the case where the weight $w$ is positive.
  The case where $w$ is negative is symmetric.
  Assume that the set
  $\{\pi \in \Delta^* \mid p(0) \xrightarrow{\pi} q(0) \wedge \weight(\pi) > 0\}$
  is not empty, and take a word $\pi$ of minimal length in that set.
  If $|\pi| < 2|Q|^3$ then we are done.
  Otherwise,
  by \cref{lem:pumping-short-cycles},
  there exists
  $r, s \in Q$, $x, d \in \setN$,
  and a factorization $\pi = \alpha \beta \gamma \theta \eta$
  satisfying the conditions of \cref{lem:pumping-short-cycles}.
  Observe that
  $
  p(0)
  \xrightarrow{\alpha \gamma \eta}
  q(0)
  $
  and $|\alpha \gamma \eta| < |\alpha \beta \gamma \theta \eta|$.
  We deduce from the minimality of $\pi$ that $\weight(\alpha \gamma \eta) \leq 0$.
  This entails that $\weight(\beta \theta) = \weight(\pi) - \weight(\alpha \gamma \eta) > 0$.
  By \cref{cor:short-oca-runs},
  since $p(0) \xrightarrow{*} r(x)$ and $s(x) \xrightarrow{*} q(0)$,
  there exists $\alpha', \eta'$ both of length at most $(|Q| + x)^3$ such that
  $p(0) \xrightarrow{\alpha'} r(x)$ and $s(x) \xrightarrow{\eta'} q(0)$.
  Similarly,
  since $r(x) \xrightarrow{*} s(x)$ via a run with no zero-test,
  there exists $\gamma'$ with no zero-test and of length $|\gamma'| \leq (|Q| + 2 x)^3$ such that
  $r(x) \xrightarrow{\gamma'} s(x)$.
  As $x \leq 2|Q|^2$,
  we get that $|\alpha' \gamma' \eta'| \leq (5|Q|^2)^3 + 2 (3|Q|^2)^3 \leq 179|Q|^6$.
  Now consider the word $\pi'$ defined by
  $\pi' \eqdef \alpha' \beta^n \gamma' \theta^n \eta'$
  where $n = 1 + |\alpha' \gamma' \eta'|$.
  Note that $p(0) \xrightarrow{\pi'} q(0)$ as
  $\gamma'$ contains no zero-test and the factorization
  $\pi = \alpha \beta \gamma \theta \eta$ satisfies the conditions of \cref{lem:pumping-short-cycles}.
  Moreover,
  $\weight(\pi') = \weight(\alpha' \gamma' \eta') + n \weight(\beta \theta)$,
  hence,
  $\weight(\pi') \geq -|\alpha' \gamma' \eta'| + n > 0$.
  We deduce from the minimality of $\pi$ that $|\pi| \leq |\pi'|$.
  It remains to show that $\pi'$ is short.
  By construction,
  $|\pi'| = |\alpha' \gamma' \eta'| + n |\beta \theta| \leq 179 |Q|^6 + 180 |Q|^6 \cdot 2|Q|^3$.
  We obtain that $|\pi'| \leq 539 |Q|^9$,
  which concludes the proof of the lemma.
\end{proof}

\begin{lemma}
  \label{lem:short-runs-with-weight-modulo}
  If $p(0) \wrightarrow{*}{w} q(0) \xrightarrow{*} p(0)$ then
  for every $m > 0$,
  there exists $\pi \in \Delta^*$ with $|\pi| < m^2 |Q|^3$ verifying
  $p(0) \xrightarrow{\pi} q(0)$ and
  $\weight(\pi) \equiv w \pmod{m}$,
  and such that
  if $w \neq \weight(\pi)$ then $q(0) \wrightarrow{*}{v} q(0)$
  for some $v \neq 0$ having the same sign as $w - \weight(\pi)$.
\end{lemma}
\begin{proof}
  Assume that $p(0) \wrightarrow{*}{w} q(0) \xrightarrow{*} p(0)$ and let $m > 0$.
  Let $C(\pi)$ denote the condition that
  if $w \neq \weight(\pi)$ then $q(0) \wrightarrow{*}{v} q(0)$
  for some $v \neq 0$ having the same sign as $w - \weight(\pi)$.
  Consider the set $S$ of all words $\pi \in \Delta^*$ such that
  $p(0) \xrightarrow{\pi} q(0)$,
  $\weight(\pi) \equiv w \pmod{m}$ and $C(\pi)$ holds.
  This set is not empty since $p(0) \wrightarrow{*}{w} q(0)$.
  Take a word $\pi$ of minimal length in $S$ and let us prove that $|\pi|$ meets the desired bound.
  Suppose, by contradiction, that $|\pi| \geq m^2 |Q|^3$.
  By \cref{lem:hill-cutting},
  there exists
  a factorization $\pi = \alpha \beta_1 \cdots \beta_m \gamma \theta_m \cdots \theta_1 \eta$,
  with $\beta_i \theta_i \neq \varepsilon$ for all $i \in \interval{1}{m}$,
  such that
  $
  p(0)
  \xrightarrow{\alpha \beta_1^{n_1} \cdots \beta_m^{n_m} \gamma \theta_m^{n_m} \cdots \theta_1^{n_1} \eta}
  q(0)
  $
  for every $n_1, \ldots, n_m \in \setN$.
  Let $u_i \eqdef \weight(\beta_1 \cdots \beta_i \theta_i \cdots \theta_1)$
  for every $i \in \interval{0}{m}$,
  with the understanding that $u_0 = \weight(\varepsilon) = 0$,
  and consider the sequence $u_0, \ldots, u_m$.
  By the pigeonhole principle,
  there exists $0 \leq i < j \leq m$ such that
  $u_i$ and $u_j$ are in the same congruence class modulo $m$.
  So $u_j = u_i + u$ for some $u \in \setZ m$.
  This means that $\weight(\beta_{i+1} \cdots \beta_j \theta_j \cdots \theta_{i+1}) = u$.
  Now,
  for each $k \in \setN$,
  let $\pi'_k$ denote the word obtained from $\pi$ by taking
  the cycles $\beta_{i+1}, \ldots, \beta_j$ and $\theta_j, \ldots, \theta_{i+1}$ exactly $k$ times,
  formally,
  $\pi'_k \eqdef \alpha \beta_1 \cdots \beta_i \beta_{i+1}^k \cdots \beta_j^k \beta_{j+1} \cdots \beta_m \gamma \theta_m \cdots \theta_{j+1} \theta_j^k \cdots \theta_{i+1}^k \theta_i\cdots \theta_1 \eta$.
  It is readily seen that $p(0) \xrightarrow{\pi'_k} q(0)$ and $\weight(\pi'_k) = \weight(\pi) + (k-1) u$.

  \smallskip

  Let us prove that $\pi'_0 \in S$.
  We have already shown that
  $p(0) \xrightarrow{\pi'_0} q(0)$ and $\weight(\pi'_0) = \weight(\pi) - u$,
  hence,
  $\weight(\pi'_0) \equiv w \pmod{m}$.
  It remains to show that $C(\pi'_0)$ holds.
  Let $s, s' \in \{-1, 0, 1\}$ denote the signs of $w - \weight(\pi)$ and $w - \weight(\pi'_0)$,
  respectively.
  If $s' = 0$ then $C(\pi'_0)$ holds trivially.
  If $s' \neq 0$ and $s = s'$ then $C(\pi'_0)$ holds because $C(\pi)$ holds.
  If $s' = 1$ and $s \leq 0$ then
  $\weight(\pi'_0) < w \leq \weight(\pi)$,
  hence, $u > 0$.
  It follows from $\weight(\pi'_k) = \weight(\pi) + (k-1) u$ that
  $p(0) \wrightarrow{*}{v} q(0)$ for infinitely many $v > 0$.
  As $q(0) \xrightarrow{*} p(0)$,
  we deduce that $q(0) \wrightarrow{*}{v} q(0)$
  for some $v > 0$,
  hence,
  $C(\pi'_0)$ holds.
  If $s' = -1$ and $s \geq 0$ then
  $\weight(\pi) \leq w < \weight(\pi'_0)$,
  hence, $u < 0$.
  It follows from $\weight(\pi'_k) = \weight(\pi) + (k-1) u$ that
  $p(0) \wrightarrow{*}{v} q(0)$ for infinitely many $v < 0$.
  As $q(0) \xrightarrow{*} p(0)$,
  we deduce that $q(0) \wrightarrow{*}{v} q(0)$
  for some $v < 0$,
  hence,
  $C(\pi'_0)$ holds.
  We have shown in all cases that $\pi'_0 \in S$.
  This contradicts the minimality of $\pi$ since
  $|\pi'_0| = |\pi| - |\beta_{i+1} \cdots \beta_j \theta_j \cdots \theta_{i+1}| < |\pi|$.
\end{proof}

We are now ready to prove the main result of this section,
namely that the reachable weights between two mutually reachable configurations
$p(0)$ and $q(0)$ can be obtained via small linear path schemes.

\begin{theorem}
  \label{lem:pumping-short-lps}
  Let $\woca = (Q, \Sigma, \Delta, \weight)$ be a WOCA.
  For every states $p, q \in Q$ and weight $w \in \setZ$ verifying
  $p(0) \wrightarrow{*}{w} q(0) \xrightarrow{*} p(0)$,
  there exists $\alpha, \beta \in \Delta^*$ and $n \in \setN$
  such that $p(0) \wrightarrow{\alpha \beta^n}{w} q(0)$,
  $q(0) \xrightarrow{\beta} q(0)$ and
  $|\alpha \beta| \leq (2|Q|)^{39}$.
\end{theorem}
\begin{proof}
  Assume that $p(0) \wrightarrow{*}{w} q(0) \xrightarrow{*} p(0)$.
  We start by fixing two short cyclic runs on $q(0)$,
  one with positive weight and one with negative weight, as follows.
  By \cref{lem:short-runs-with-positive-weight},
  if $q(0) \wrightarrow{*}{w} q(0)$ for some $w > 0$ then
  $q(0) \xrightarrow{\beta} q(0)$ for some
  $\beta \in \Delta^*$ such that $\weight(\beta) > 0$ and $|\beta| \leq 539|Q|^9$.
  Let $\beta \eqdef \varepsilon$ otherwise.
  Analogously,
  if $q(0) \wrightarrow{*}{w} q(0)$ for some $w < 0$ then
  $q(0) \xrightarrow{\theta} q(0)$ for some
  $\theta \in \Delta^*$ such that $\weight(\theta) < 0$ and $|\theta| \leq 539|Q|^9$.
  Let $\theta \eqdef \varepsilon$ otherwise.

  \smallskip

  By \cref{lem:short-runs-with-weight-modulo},
  for each $m \in \{1, \weight(\beta), -\weight(\theta), -\weight(\beta)\weight(\theta)\}$
  such that $m > 0$,
  there exists $\alpha_m \in \Delta^*$ with $|\alpha_m| < m^2 |Q|^3$ verifying
  $p(0) \xrightarrow{\alpha_m} q(0)$ and
  $\weight(\alpha_m) \equiv w \pmod{m}$,
  and such that
  if $w \neq \weight(\alpha_m)$ then $q(0) \wrightarrow{*}{v} q(0)$
  for some $v \neq 0$ having the same sign as $w - \weight(\alpha_m)$.
  Let $u_m \eqdef w - \weight(\alpha_m)$ and note that $u_m \in \setZ m$.
  Moreover, $u_m > 0$ implies $\beta \neq \varepsilon$
  and $u_m < 0$ implies $\theta \neq \varepsilon$.
  We now consider four cases depending on the emptiness of $\beta$ and $\theta$.

  \smallskip

  If $\beta = \theta = \varepsilon$ then we use $\alpha_m$ and $u_m$ for $m \eqdef 1$.
  We derive from $\beta = \theta = \varepsilon$ that $u_m = 0$.
  It follows that $w = \weight(\alpha_m) + u_m = \weight(\alpha_m \beta)$.

  If $\beta \neq \varepsilon$ and $\theta = \varepsilon$ then we use $\alpha_m$ and $u_m$
  for $m \eqdef \weight(\beta) > 0$.
  We derive from $\theta = \varepsilon$ that $u_m \geq 0$.
  As $u_m \in \setZ m$, we get that $u_m = n \weight(\beta)$ for some $n \in \setN$.
  It follows that $w = \weight(\alpha_m) + u_m = \weight(\alpha_m \beta^n)$.

  If $\beta = \varepsilon$ and $\theta \neq \varepsilon$ then we use $\alpha_m$ and $u_m$
  for $m \eqdef -\weight(\theta) > 0$.
  We derive from $\beta = \varepsilon$ that $u_m \leq 0$.
  As $u_m \in \setZ m$, we get that $u_m = n \weight(\theta)$ for some $n \in \setN$.
  It follows that $w = \weight(\alpha_m) + u_m = \weight(\alpha_m \theta^n)$.

  If $\beta \neq \varepsilon$ and $\theta \neq \varepsilon$ then we use $\alpha_m$ and $u_m$
  for $m \eqdef - \weight(\beta) \weight(\theta) > 0$.
  As $u_m \in \setZ m$, we get that $u_m = n \weight(\beta) \weight(\theta)$ for some $n \in \setZ$.
  If $n \leq 0$ then $w = \weight(\alpha_m) + u_m = \weight(\alpha_m \beta^{n \weight(\theta)})$.
  If $n \geq 0$ then $w = \weight(\alpha_m) + u_m = \weight(\alpha_m \theta^{n \weight(\beta)})$.

  \smallskip

  We have shown in each case that $w = \weight(\alpha_m \gamma^n)$ for some
  $m \in \{1, \weight(\beta), -\weight(\theta), -\weight(\beta)\weight(\theta)\}$
  with $m > 0$,
  some $\gamma \in \{\beta, \theta\}$ and some $n \in \setN$.
  Moreover,
  our choice of $\beta$ and $\theta$ ensures that
  $m \leq (539|Q|^9)^2$,
  $q(0) \xrightarrow{\gamma} q(0)$ and $|\gamma| \leq 539|Q|^9$.
  Recall that $p(0) \xrightarrow{\alpha_m} q(0)$ and $|\alpha_m| < m^2 |Q|^3$.
  It follows that
  $p(0) \wrightarrow{\alpha_m \gamma^n}{w} q(0)$
  and
  $|\alpha_m| \leq (539|Q|^9)^4 |Q|^3$.
  We obtain that $|\alpha_m \gamma| \leq (2|Q|)^{39}$,
  which concludes the proof of the theorem.
\end{proof}

\section{Succinct Flattenability of 2-TVASS}
\label{sec:succint-flattenability-of-2-tvass}
We have shown in \cref{sec:flattenability-of-2-tvass} that $2$-TVASS are flattenable.
We now prove that flattenability of $2$-TVASS can be witnessed by small linear path schemes.

\smallskip

We first introduce a binary relation $\rightsquigarrow$ over the states of
a $2$-TVASS $\vass$, defined by $p \rightsquigarrow q$ if $p(0,x)\xrightarrow{*}q(0,y)$ for some
$x,y\in\setN$. Notice that this relation is transitive since
$p(0,x)\xrightarrow{*}q(0,y)$ and $q(0,x')\xrightarrow{*}r(0,y')$
implies $p(0,x+x')\xrightarrow{*}q(0,y+x')\xrightarrow{*}r(0,y+y')$ by
monotonicity. As mentioned in \cref{sec:weighted-one-counter-automata}, a WOCA can be associated to any
$2$-TVASS by considering actions on the second counter, the one that is not tested
for zero, as weights. Under this observation, $p \rightsquigarrow q$ if, and only if,
$p(0)\xrightarrow{*}q(0)$ in the associated WOCA.
This observation also provides a way to convert
\cref{lem:pumping-short-lps} to the following lemma.
\begin{lemma}
  \label{lem:largelarge}
  There exists a constant $h\geq 1$ such that,
  for every $2$-TVASS $\vass=(Q,\Sigma,\Delta)$,
  if $p(0,x)\xrightarrow{*}q(0,y)$ with $x,y\geq (|Q|+\norm{\Sigma})^h$ and $q
  \rightsquigarrow p$, then there exists a linear path scheme $L$ with
  $|L| \leq (|Q|+\norm{\Sigma})^{O(1)}$ and $|L|_*=1$ such that $p(0,x)\xrightarrow{L}q(0,y)$.
\end{lemma}
\begin{proof}
  A $2$-TVASS cannot be directly translated into a WOCA since some addition transitions $(p,\vec{a},q)$ may satisfy $\norm{\vec{a}}>1$. However, by introducing intermediate states and transitions between $p$ and $q$, we can overcome this problem. It follows that we can assume, without loss of generality, that every addition transition $(p,\vec{a},q)$ satisfies $\norm{\vec{a}}\leq 1$. Additionally, by introducing for each state $p$ and each addition transition $\delta$ an intermediate state, we can assume that if $(p,\vec{a},q)$ and $(p,\vec{b},q)$ are two addition transitions such that $\vec{a}(1)=\vec{b}(1)$ then $\vec{a}(2)=\vec{b}(2)$. Thanks to this assumption, we can associate to a $2$-TVASS $\vass=(Q,\Sigma,\Delta)$ a WOCA $(Q,\Sigma',\Delta',\weight)$ in such a way $(p,(a,b),q)$ is a transition in $\vass$ if, and only if, $(p,a,q)$ is a transition in the WOCA weighted by $b$, and such that $(p,\ztest,q)$ is a transition in $\vass$ if, and only if, $(p,\ztest,q)$ is a transition in the WOCA, and in that case the transition is weighted by zero. Now, let us consider two configurations $p(0,x)$ and $q(0,y)$ with $x,y\geq (2|Q|)^{39}$ such that $p(0,x)\xrightarrow{*}q(0,y)$ and $q \rightsquigarrow p$ in $\vass$.
  Notice that $p(0)\wrightarrow{*}{w}q(0)$ and $q(0) \xrightarrow{*}p(0)$ in the
  WOCA with $w=y-x$. From
  \cref{lem:pumping-short-lps}, it follows that there exists a path
  $\pi$ from $p$ to $q$, a cycle $\theta$ on $q$, and $n\in\setN$ such that
  $p(0)\wrightarrow{\pi\theta^n}{w}q(0)$ with
  $|\pi\theta|\leq (2|Q|)^{39}$.
  Notice that $\pi$ and $\theta$ in the WOCA corresponds to a path $\alpha$ and a cycle $\beta$ in $\vass$, respectively. Since $x,y\geq (2|Q|)^{39}\geq |\alpha\beta|$, observe that $p(0,x)\xrightarrow{\alpha\beta^n}q(0,y)$ since each execution of $\beta$ can decrease or increase the second counter by a value bounded by $|\beta|\leq (2|Q|)^{39}$.
\end{proof}

The previous lemma captures the reachability relation of a $2$-TVASS
between configurations $p(0,x)$ and $q(0,y)$ with $p \rightsquigarrow q \rightsquigarrow p$ and $x,y$
are large. In order to capture the same relation when $x$ or $y$
are small, the following result will be useful.
\begin{theorem}[\cite{DBLP:conf/lics/EnglertLT16}]\label{thm:ranko}
  For every $2$-VASS $\vass = (Q, \Sigma, \Delta)$,
  and for every configurations $p(\vec{x})$ and $q(\vec{y})$
  such that $p(\vec{x}) \xrightarrow{*} q(\vec{y})$ in $\vass$,
  there exists a path $\pi$ such that $p(\vec{x}) \xrightarrow{\pi} q(\vec{y})$
  and satisfying
  $$
  |\pi|\leq (|Q|+ \norm{\Sigma}+\norm{\vec{x}}+\norm{\vec{y}})^{O(1)}
  \:.
  $$
\end{theorem}

We are now ready to refine \cref{lem:increasing-vloop-flattenable} with complexity bounds.
Recall that $\vloop_q$ is the vertical loop relation on $q$ defined by
$
{\vloop_q} = \{(q(0, x), q(0, y)) \mid q(0, x) \xrightarrow{*} q(0, y)\}
$.

\begin{lemma}\label{lem:q2q}
  For every $2$-TVASS $\vass=(Q,\Sigma,\Delta)$ and state $q \in Q$,
  we have ${\vloop_q} \subseteq \bigcup_{L \in \Lambda} {\xrightarrow{L}}$
  for some finite set $\Lambda$ of linear path schemes $L$ such that
  $|L|\leq (|Q|+\norm{\Sigma})^{O(1)}$ and $|L|_*\leq O(|Q|^2)$.
\end{lemma}
\begin{proof}
  Let $h\geq 1$ be the constant of \cref{lem:largelarge}, and let $c\geq 1$ be a constant satisfying $O(1)\leq c$ and $O(|Q|^2)\leq c|Q|^2$ in \cref{lem:largelarge}, \cref{thm:2-vass-flat}, and \cref{thm:ranko}. Let us consider a $2$-TVASS $\vass=(Q,\Sigma,\Delta)$ and let $N=|Q|+\norm{\Sigma}$.

  Observe that a run from a configuration
  $q(0,x)$ to a configuration $q(0,y)$ can be split in such a way:
  $$q(0,x)=q_1(0,x_1) \xrightarrow{A^* \,\cup\, T}q_2(0,x_2) \cdots \xrightarrow{A^* \,\cup\, T}q_k(0,x_k)=q(0,y)$$
  Moreover, by removing some parts of such a run, we can assume
  that the configurations $q_j(0,x_j)$ are pairwise distinct.

  Notice that if $x_j< N^h$ for every $j\in\interval{1}{k}$, then
  $k\leq |Q| \cdot N^h\leq N^{h+1}$. By applying
  \cref{thm:ranko}, we deduce that for every $j \in \interval{2}{k}$, there exists
  a path $\alpha_j$ such that
  $q_{j-1}(0,x_{j-1})\xrightarrow{\alpha_j}q_j(0,x_j)$ with
  $|\alpha_j|\leq (N+2N^h)^c$. In particular $\alpha$ defined as
  $\alpha_2\cdots\alpha_k$ is a path such that
  $q(0,x)\xrightarrow{\alpha}q(0,y)$ with $|\alpha|\leq
  k \cdot (N+2N^h)^c\leq N^e$ for some constant $e$.
  We are
  done with the linear path scheme $L=\alpha$.
  So, we can assume that there exists $j$ such that $x_j\geq
  N^h$. In that case, we introduce $j_\textsf{min}$ and $j_\textsf{max}$
  respectively defined as the minimal and the maximal $j$ satisfying
  this property.

  Let us prove that there exists a linear path scheme
  $L_\textsf{min}$ such that $|L_\textsf{min}|\leq N^e +N^c$ and
  $|L_\textsf{min}|_*\leq c|Q|^2$ and such that
  $q(0,x)\xrightarrow{L_\textsf{min}}q_{j_\textsf{min}}(0,x_{j_\textsf{min}})$.
  Observe that if $j_\textsf{min}=1$, the proof is immediate with
  $L_\textsf{min}$ reduced to the empty path. 
  If $j_\textsf{min}>1$, as $x_{j}<N^h$ for every $1\leq j<j_\textsf{min}$, we
  deduce from the previous paragraph that there exists a path
  $\alpha_\textsf{min}$ with a length bounded by
  $N^e$ such that
  $q(0,x)\xrightarrow{\alpha_\textsf{min}}q_{j_\textsf{min}-1}(0,x_{j_\textsf{min}-1})$. Recall
  that
  $q_{j_\textsf{min}-1}(0,x_{j_\textsf{min}-1})\xrightarrow{A^*\cup T}q_{j_\textsf{min}}(0,x_{j_\textsf{min}})$.
  Based
  on \cref{thm:2-vass-flat}, we deduce that there exists a linear
  path scheme $L_0$ such that
  $q_{j_\textsf{min}-1}(0,x_{j_\textsf{min}-1})\xrightarrow{L_0}q_{j_\textsf{min}}(0,x_{j_\textsf{min}})$
  with $|L_0|\leq N^c$ and $|L_0|_*\leq c|Q|^2$. By
  considering $L_\textsf{min}=\alpha_\textsf{min}L_0$ we are done.
  Symmetrically, there exists a linear path scheme
  $L_\textsf{max}$ such that $|L_\textsf{max}|\leq N^e +N^c$ and
  $|L_\textsf{max}|_*\leq c|Q|^2$ and such that
  $q_{j_\textsf{max}}(0,x_{j_\textsf{max}}) \xrightarrow{L_\textsf{max}}q(0,y)$.

  Note that $q_{j_\textsf{max}} \rightsquigarrow q_k = q_1 \rightsquigarrow q_{j_\textsf{min}}$,
  hence, $q_{j_\textsf{max}} \rightsquigarrow q_{j_\textsf{min}}$.
  By applying \cref{lem:largelarge}, we
  deduce that there exists a linear path scheme $L_1$ with $|L_1|\leq N^c$
  and $|L|_*=1$ such that
  $q_{j_\textsf{min}}(0,x_{j_\textsf{min}})\xrightarrow{L_1}q_{j_\textsf{max}}(0,x_{j_\textsf{max}})$,
  It
  follows that the linear path scheme $L$ defined as
  $L_\textsf{min}L_1L_\textsf{max}$ satisfies the lemma.
\end{proof}

\begin{corollary}\label{cor:lps}
  Every $2$-TVASS is flattenable. Furthermore, for every configurations $p(\vec{x})$ and $q(\vec{y})$ of a $2$-TVASS $\vass=(Q,\Sigma,\Delta)$ such that $p(\vec{x})\xrightarrow{*}q(\vec{y})$, there exists a linear path scheme $L$ with $|L|\leq
  (|Q|+\norm{\Sigma})^{O(1)}$
  and $|L|_*\leq O(|Q|^3)$ such that 
  $p(\vec{x})\xrightarrow{L}q(\vec{y})$.
\end{corollary}
\begin{proof}
  The proof is a direct corollary of \cref{lem:vloop-flatness-to-flatness}, \cref{thm:2-vass-flat}, and \cref{lem:q2q}.
\end{proof}

\begin{example}
  As an illustration of \cref{cor:lps},
  let us continue \cref{ex:2TVASS,ex:LPS} and
  provide a finite set $\Lambda$ of ``small'' linear path schemes such that
  $A(\vec{x}) \xrightarrow{*} B(\vec{y})$ if, and only if,
  $A(\vec{x}) \xrightarrow{L} B(\vec{y})$ for some $L \in \Lambda$.
  First, we observe that for every $x, y \in \setN$,
  if $A(0, x) \xrightarrow{*} A(0, y)$ then
  $x = y$ or the following condition is satisfied:
  $$
  x \geq 2 \ \wedge \ y \geq x+2 \ \wedge \ (y = x+3 \Rightarrow x \geq 5) \ \wedge \ (y = x+5 \Rightarrow x \geq 3)
  \:.
  $$
  Second, we introduce the paths
  $\pi = \delta_{AB} \delta_{BB} \delta_{BB} \delta_{BA} \delta_{AA}$
  and $\sigma = \delta_{AB} (\delta_{BB})^5 \delta_{BA} (\delta_{AA})^2$.
  It is routinely checked that
  $A(0, x) \xrightarrow{\pi} A(0, y)$ if, and only if,
  $x \geq 2$ and $y = x + 2$.
  Similarly,
  $A(0, x) \xrightarrow{\sigma} A(0, y)$ if, and only if,
  $x \geq 5$ and $y = x + 3$.
  We derive that
  $A(0, x) \xrightarrow{*} A(0, y)$  if, and only if,
  $A(0, x) \xrightarrow{\pi^* \cdot \{\varepsilon, \sigma\}} A(0, y)$.
  We are now done by taking $\Lambda = \{L_1, L_2\}$ where
  $L_1$ and $L_2$ are the linear path schemes defined by
  $L_1 = (\delta_{AA})^* \cdot \pi^* \cdot \delta_{AB} \cdot (\delta_{BB})^*$ and
  $L_2 = (\delta_{AA})^* \cdot \pi^* \cdot \sigma \delta_{AB} \cdot (\delta_{BB})^*$.
  \qed
\end{example}

\section{Linear Path Schemes to Systems of Equations}
\label{sec:linear-path-scheme-to-equations}
In this section, we associate to a linear path scheme $L=\alpha_0\beta_1^*\alpha_1\cdots \beta_k^*\alpha_k$ of a $d$-TVASS $\vass$ from a state $p$ to a state $q$, and to a vectors $\vec{x},\vec{y}\in\setN^d$, a system of linear inequalities $S_{\vec{x},L,\vec{y}}$ encoding over the variables $(n_1,\ldots,n_k)$ the following constraint:
$$p(\vec{x})\xrightarrow{\alpha_0\beta_1^{n_1}\alpha_1\cdots \beta_k^{n_k}\alpha_k}q(\vec{y})$$
Such a system is classical for $d$-VASS, but for $d$-TVASS, the presence of zero-test transitions in the linear path scheme $L$ requires some additional work.

\smallskip

Let us first characterize the binary relation $\xrightarrow{\pi}$ thanks to a system of linear inequalities associated to a path $\pi$. We introduce the \emph{displacement} $\disp{\delta}$ of a transition $\delta$ as the vector in $\setZ^d$ defined by $\disp{\delta}\eqdef\vec{a}$ if $\delta$ is of the form $(p, \vec{a}, q)$ with $\vec{a} \in \setZ^d$ and $\disp{\delta}\eqdef\vec{0}$ if $\delta$ is of the form $(p, \ztest, q)$. The \emph{displacement} of a path $\pi=\delta_1\ldots\delta_n$ is $\disp{\pi}\eqdef\disp{\delta_1} + \cdots + \disp{\delta_n}$. We also introduce the vector $\vec{m}_\pi\in\setN^d$ defined component-wise for every $i\in\interval{1}{d}$ by $\vec{m}_\pi(i)\eqdef\max_\alpha(-\disp{\alpha}(i))$ where $\alpha$ ranges over the prefixes of $\pi$.

\smallskip

A path $\pi$ from a state $p$ to a state $q$ is said to be \emph{feasible} if $p(\vec{x})\xrightarrow{\pi}q(\vec{y})$ for some $\vec{x},\vec{y}\in\setN^d$.
We introduce the partial order $\geq_1$ defined over $\setN^d$ by $\vec{x}\geq_1\vec{y}$ if $\vec{x}(1)=\vec{y}(1)$ and $\vec{x}(i)\geq \vec{y}(i)$ for every $i\in\interval{2}{d}$.
We let $\succeq_\pi$ denote the partial order over $\setN^d$ defined as follows:
$\succeq_\pi$ is $\geq_1$ if $\pi$ contains a zero-test transition, and $\succeq_\pi$ is $\geq$ otherwise.
\begin{restatable}{lemma}{pathinTVASS}
  \label{lem:path-in-TVASS}
  Let $\pi$ be a feasible path from a state $p$ to a state $q$. For every $\vec{x},\vec{y}\in\setN^d$, we have:
  $$p(\vec{x})\xrightarrow{\pi}q(\vec{y})~~\Longleftrightarrow~~\vec{x}\succeq_\pi\vec{m}_\pi\ \wedge \ \vec{y}=\vec{x}+\disp{\pi}$$
\end{restatable}

Let us recall that in \cref{sec:flattenability-of-2-tvass} we introduce the $d$-TVASS $\overline{\vass}$ obtained from $\vass$ by
reversing the effect of each transition, i.e.,
$\overline{\Delta} = \{\overline{\delta} \mid \delta\in\Delta\}$ where
$\overline{(p,\vec{a},q)}=(q,-\vec{a},p)$ and $\overline{(p,T,q)}=(q,T,p)$. Given a path $\pi=\delta_1\cdots\delta_n$ from $p$ to $q$ in $\vass$, we introduce the path $\overline{\pi}$ from $q$ to $p$ in $\overline{\vass}$ defined as $\overline{\pi}\eqdef \overline{\delta_n}\cdots\overline{\delta_1}$. Observe that $p(\vec{x})\xrightarrow{\pi}q(\vec{y})$ if, and only if, $q(\vec{y})\xrightarrow{\overline{\pi}}p(\vec{x})$.
\begin{lemma}\label{lem:alphaalpha}
  We have $\vec{m}_{\overline{\pi}}=\vec{m}_\pi+\disp{\pi}$.
\end{lemma}
\begin{proof}
  Observe that for any decomposition of $\pi$ into $\alpha\alpha'$, we have $\disp{\pi}=\disp{\alpha}+\disp{\alpha'}$. Hence $-\disp{\alpha}+\disp{\pi}=\disp{\alpha'}=-\disp{\overline{\alpha}'}$. In particular $\max_\alpha(-\disp{\alpha}(i)+\disp{\pi}(i))=\max_{\alpha'}(-\disp{\overline{\alpha}'}(i))$ for every $i\in\interval{1}{d}$ where $\alpha$ ranges over the prefixes of $\pi$ and $\alpha'$ over the suffixes of $\pi$. By observing that $\overline{\alpha}'$ ranges over all the prefixes of $\overline{\pi}$ when $\alpha'$ ranges over the suffixes of $\pi$, we get $\vec{m}_\pi+\disp{\pi}=\vec{m}_{\overline{\pi}}$.
\end{proof}

We are now ready to express the relation $\xrightarrow{\beta^n}$ where $\beta$ is a cycle on a state $q$ and $n\geq 1$ is a positive natural number.
\begin{restatable}{lemma}{cycleinTVASS}
  \label{lem:cycle-in-TVASS}
  Let $\beta$ be a feasible cycle on a state $q$. For every $\vec{x},\vec{y}\in\setN^d$ and $n\in\setN\setminus\{0\}$, we have:
  $$q(\vec{x})\xrightarrow{\beta^n}q(\vec{y}) ~~\Longleftrightarrow~~\vec{x}\succeq_\pi\vec{m}_\beta \ \wedge \ \vec{y}\succeq_\pi \vec{m}_{\overline{\beta}} \ \wedge \ \vec{y}=\vec{x}+n\disp{\beta}$$
\end{restatable}

A linear path scheme $L=\alpha_0\beta_1^*\alpha_1\cdots \beta_k^*\alpha_k$ is said to be \emph{feasible} if the paths $\alpha_0,\ldots,\alpha_k$ and the cycles $\beta_1,\ldots,\beta_k$ are feasible. We are now ready to introduce a system of linear inequalities $S_{\vec{x},L,\vec{y}}$ over the variables $(n_1,\ldots,n_k)$ where $\vec{x},\vec{y}\in \setN^d$, and $n_1,\ldots,n_k$ are variables ranging over $\setN$ as follows:
$$\bigwedge_{j=0}^k \vec{x}_j\succeq_{\alpha_j} \vec{m}_{\alpha_j}
  \quad\wedge\quad
  \bigwedge_{j=1}^k \vec{y}_{j-1}\succeq_{\beta_j} \vec{m}_{\beta_j}\wedge \vec{x}_j\succeq_{\beta_j} \vec{m}_{\overline{\beta_j}}
  \quad\wedge\quad  \vec{y}=\vec{y}_k$$
  where
  $\vec{x}_0$ is the expression $\vec{x}$, and by induction over $j$, by letting $\vec{y}_j$ be the expression $\vec{x}_j+\disp{\alpha_j}$ for every $j \in \interval{0}{k}$, and $\vec{x}_{j}$ is the expression $\vec{y}_{j-1}+n_j\disp{\beta_j}$ for every $j \in \interval{1}{k}$. From \cref{lem:path-in-TVASS,lem:cycle-in-TVASS}, we derive the following corollary.
\begin{corollary}\label{cor:system}
  Assume that $L=\alpha_0\beta_1^*\alpha_1\cdots \beta_k^*\alpha_k$ is a feasible linear path scheme from a state $p$ to a state $q$, and let $\vec{x},\vec{y}\in\setN^d$. If $(n_1,\ldots,n_k)$ is a solution of $S_{\vec{x},L,\vec{y}}$ then
  $$p(\vec{x})\xrightarrow{\alpha_0\beta_1^{n_1}\alpha_1\cdots\beta_k^{n_k}\alpha_k}q(\vec{y})$$
  Conversely, a tuple $(n_1,\ldots,n_k)$ with $n_1,\ldots,n_k\geq 1$ that satisfies the previous relation is a solution of $S_{\vec{x},L,\vec{y}}$.
\end{corollary}

\begin{remark}
  We can easily extend the definition of $S_{\vec{x},L,\vec{y}}$ to encode linear path schemes of extended $d$-TVASS with zero-test actions on any counter.
\end{remark}

\section{Complexity Results}
\label{sec:complexity-results}
This section utilizes the results of
\cref{sec:succint-flattenability-of-2-tvass,sec:linear-path-scheme-to-equations}
to characterize the complexity of the reachability problem in $2$-TVASS.
We assume that $2$-TVASS and their configurations are encoded in binary, and that
sizes are defined as expected (up to a polynomial).
Under this binary encoding, the reachability problem in $2$-TVASS is shown to be PSPACE-complete.
Since the reachability problem for $2$-VASS (i.e., without zero-test transitions)
is already PSPACE-hard~\cite{DBLP:journals/iandc/FearnleyJ15},
we only need to prove PSPACE-membership.

\smallskip

The PSPACE complexity upper-bound is obtained via small solutions of the system of inequalities $S_{\vec{x},L,\vec{y}}$ associated to a linear path scheme $L$ and a pair of vectors $\vec{x}, \vec{y} \in \setN^2$.
We first recall some results about small solution of systems of equations.
\begin{theorem}[\cite{pottier91}]\label{thm:pottier}
  Let $M=(M_{i,j})$ be a matrix in $\setZ^{e\times k}$. Every solution $\vec{x}\in\setN^k$ of $M\vec{x}=\vec{0}$ is a finite sum of solutions $\vec{y}\in\setN^k$ satisfying additionally $\sum_{i=1}^k\vec{y}(i)\leq (1+m)^k$ where $m=\max_i\sum_{j=1}^k|M_{i,j}|$.
\end{theorem}
We also recall the classical application of the previous theorem to systems of inequalities with constant terms (the vector $\vec{b}$ in the following corollary).
\begin{corollary}\label{cor:pottier}
  Let $M=(M_{i,j})$ be a matrix in $\setZ^{e\times k}$ and let $\vec{b}\in\setZ^e$. If there exists a solution $\vec{x}\in\setN^k$ of $M\vec{x}\geq \vec{b}$ then there exists a solution $\vec{y}\in\setN^k$ such that $\sum_{i=1}^k\vec{y}(i)\leq (2+m)^{k+1+e}$ where $m=\max_i\sum_{j=1}^k|M_{i,j}|+|\vec{b}(i)|$.
\end{corollary}
\begin{proof}
  Let $\vec{y}$ be the vector in $\setN^e$ defined as $\vec{y}\eqdef M\vec{x}-\vec{b}$ and observe that $(\vec{x},1,\vec{y})$ is a solution of $M\vec{x}-t\vec{b}-\vec{y}=\vec{0}$. From \cref{thm:pottier} we derive that $(\vec{x},1,\vec{y})$ can be decomposed as a finite sum of ``small solutions'' $(\vec{u},s,\vec{v})$ with $\sum_{j=1}^k\vec{u}(j)+s+\sum_{i=1}^e\vec{v}(i)\leq (2+m)^{k+1+e}$. Since the sum of those small solutions is $1$ on the ``$s$'' component, exactly one of them is $1$ on that component. This solution $(\vec{u},1,\vec{v})$ provides a vector $\vec{u}$ with $\sum_{j=1}^k\vec{u}(j)\leq (2+m)^{k+1+e}$ such that $M\vec{u}\geq \vec{b}$.
\end{proof}

We deduce a bound on minimal runs between two configurations.
\begin{lemma}\label{lem:shortpath}
 For every configurations $p(\vec{x})$ and $q(\vec{y})$ of a $2$-TVASS $\vass = (Q, \Sigma, \Delta)$ such that
 $p(\vec{x})\xrightarrow{*}q(\vec{y})$, there exists a path $\pi$ such that
 $p(\vec{x})\xrightarrow{\pi}q(\vec{y})$ and such that:
  $$|\pi|\leq (|Q|+\norm{\vec{x}}+\norm{\vec{y}}+\norm{\Sigma})^{O(|Q|^3)}$$
\end{lemma}
\begin{proof}
  Let $c\geq 1$ be a constant satisfying \cref{cor:lps}, i.e., such that $O(1)\leq c$ and $O(|Q|^3)\leq c|Q|^3$.
  Consider a $2$-TVASS $\vass$ and let $p(\vec{x})$ and $q(\vec{y})$ be two configurations such that $p(\vec{x})\xrightarrow{*}q(\vec{y})$.
  The case where $\Sigma \subseteq \{\vec{0}, \ztest\}$ is trivial (there is a run of length at most $|Q|$ in that case),
  so we assume that $\norm{\Sigma} \geq 1$ for the remainder of the proof.
  Let us introduce $N=|Q|+\norm{\Sigma}$. \cref{cor:lps} shows that there exists a linear path scheme $L=\alpha_0\beta_1^*\alpha_1\cdots \beta_k^*\alpha_k$ with $p(\vec{x})\xrightarrow{L}q(\vec{y})$ and such that $|L|\leq N^c$ and $k\leq c|Q|^3$. It follows that there exists $n_1,\ldots,n_k\in\setN$ such that:
  $$p(\vec{x})\xrightarrow{\alpha_0\beta_1^{n_1}\alpha_1\cdots \beta_k^{n_k}\alpha_k}q(\vec{y})$$
  By removing from $L$ the cycles $\beta_j$ such that $n_j=0$, we can assume, without loss of generality, that $n_j\geq 1$ for every $j \in \interval{1}{k}$. It follows that $L$ is feasible. From \cref{cor:system} we deduce that $(n_1,\ldots,n_k)$ is a solution of $S_{\vec{x},L,\vec{y}}$. From \cref{cor:pottier} we deduce that there exist $m_1,\ldots,m_k\in\setN$ such that $(m_1,\ldots,m_k)$ satisfies $S_{\vec{x},L,\vec{y}}$ and such that $m_1+\cdots+m_k\leq (2+m)^{k+1+e}$ where $m\leq \norm{\vec{x}}+\norm{\vec{y}}+\norm{\Sigma} {\cdot} |L|$ and $e\eqdef 9k+7$. This expression for $e$ comes from the encoding of $\geq_1$ with $3$ inequalities. Let us introduce the path $\pi=\alpha_0\beta_1^{m_1}\alpha_1\cdots \beta_k^{m_k}\alpha_k$ and observe that $p(\vec{x})\xrightarrow{\pi}q(\vec{y})$ from \cref{cor:system}. Moreover $|\pi|$ is bounded by:
  $$
    (2+m)^{10k+8}{\cdot}|L|
    \leq (2+\norm{\vec{x}}+\norm{\vec{y}}+\norm{\Sigma} N^c)^{10c|Q|^3+8}N^c
    \leq (|Q|+\norm{\vec{x}}+\norm{\vec{y}}+\norm{\Sigma})^{O(|Q|^3)}
    $$
    This concludes the proof of the lemma.
\end{proof}

We are now ready to characterize the complexity of the reachability problem in $2$-TVASS.
This decision problem asks,
given a $2$-TVASS $\vass = (Q, \Sigma, \Delta)$ and two configurations $p(\vec{x})$ and $q(\vec{y})$,
whether $p(\vec{x}) \xrightarrow{*} q(\vec{y})$.
By \cref{lem:shortpath},
if $p(\vec{x}) \xrightarrow{*} q(\vec{y})$ then
there exists a run from $p(\vec{x})$ to $q(\vec{y})$ of length at most
exponential in the sizes of $\vass$, $p(\vec{x})$ and $q(\vec{y})$.
Notice that configurations along that run have a polynomial size
(with respect to the size of the input problem).
It follows that a polynomial-space bounded exploration of the reachability set
provides a way to decide the reachability problem.
We have shown the following theorem.
\begin{theorem}
  The reachability problem for $2$-TVASS is PSPACE-complete.
\end{theorem}

\begin{remark}
  Other natural problems on $2$-TVASS are PSPACE-complete.
  In \cref{appendix:sec:complexity-results},
  we derive from the succinct flattenability of $2$-TVASS that
  the boundedness problem and the termination problem are both decidable
  in polynomial space.
  These results are obtained by providing a polynomial bound on the size
  of reachable configurations of a bounded $2$-TVASS.
\end{remark}

\section{Conclusion and Perspectives}
\label{sec:conclusion}
We have shown in this paper that extending $2$-VASS with zero-tests on the first counter
is for free, in the sense that the reachability problem remains PSPACE-complete
(and so do the boundedness and termination problems).
As in the case of $2$-VASS,
a crucial step in our approach is what we call \emph{succinct flattenability}, i.e.,
the existence of small linear path schemes witnessing flattenability.
Succinct flattenability of $2$-VASS was leveraged by Englert et al. in~\cite{DBLP:conf/lics/EnglertLT16}
to show that reachability in $2$-VASS is NL-complete when the input integers are encoded in unary.
The question whether reachability in unary $2$-TVASS remains NL-complete is left open.
We conjecture that this question can be answered positively,
by leveraging our succinct flattenability result for $2$-TVASS and
by extending~\cite{DBLP:conf/lics/EnglertLT16} with zero-tests on the first counter.

\bibliographystyle{plainurl}
\bibliography{biblio}

\clearpage
\appendix

\section{Missing Proofs for \cref{sec:flattenability-of-2-tvass}}
\label{appendix:sec:flattenability-of-2-tvass}
\vloopflatnesstoflatness*
\begin{proof}
  Consider a run $\rho$ from a configuration $p(\vec{x})$ to a configuration $q(\vec{y})$.
  If $\rho$ contains no zero-test transition,
  then $p(\vec{x}) \xrightarrow{A^*} q(\vec{y})$ and we are done.
  Otherwise,
  by splitting $\rho$ at configurations where the first counter is zero,
  we obtain that $\rho$ may be written as
  $$
  p(\vec{x})
  \xrightarrow{A^* \,\cup\, T} s_1(0, z_1)
  \xrightarrow{A^* \,\cup\, T} s_2(0, z_2)
  \cdots
  \xrightarrow{A^* \,\cup\, T} s_n(0, z_n)
  \xrightarrow{A^* \,\cup\, T} q(\vec{y})
  $$
  with $n \geq 1$.
  The sequence $s_1 \cdots s_n$, viewed as a nonempty word in $Q^*$,
  may be factorized as a concatenation of at most $|Q|$ words
  that each start and end with the same letter.
  So our run $\rho$ may be factorized as
  $$
  p(\vec{x})
  \xrightarrow{A^* \,\cup\, T} q_1(0, x_1)
  \xrightarrow{*} q_1(0, y_1)
  \cdots
  \xrightarrow{A^* \,\cup\, T} q_k(0, x_k)
  \xrightarrow{*} q_k(0, y_k)
  \xrightarrow{A^* \,\cup\, T} q(\vec{y})
  $$
  with $1 \leq k \leq |Q|$.
  This concludes the proof of the lemma.
\end{proof}

\increasingvloopflattenable*
\begin{proof}
  Assume that $q(0, x) \xrightarrow{*} q(0, y)$ with $y > x$.
  The difference $d = y - x$ is in $D_x$.
  We first consider the case where $x \leq t$.
  Recall that $D_x = F_x \cup (B_x + \setN m)$.
  If $d = f$ for some $f \in F_x$ then
  $q(0, x) \xrightarrow{\alpha_{x, f}} q(0, x + f) = q(0, y)$
  and we are done since $\alpha_{x, f} \in \Lambda_x$.
  Otherwise,
  we have $d = b + km$ for some $b \in B_x$ and $k \in \setN$.
  Recall that $q(0, h) \xrightarrow{\beta} q(0, h + m)$ and
  note that $x + b \geq b \geq h$.
  We get
  $q(0, x) \xrightarrow{\alpha_{x, b}} q(0, x + b) \xrightarrow{\beta^k} q(0, x + b + km) = q(0, y)$
  and we are done since $\alpha_{x, b} \beta^* \in \Lambda_x$.

  \smallskip

  The other case is when $x > t$.
  In that case, $d = y - x$ is in $D_t$ as $D_x = D_t$, hence,
  $q(0, t) \xrightarrow{*} q(0, t + d)$.
  We get from the above case that $q(0, t) \xrightarrow{L} q(0, t + d)$
  for some $L \in \Lambda_t$.
  It follows that $q(0, x) \xrightarrow{L} q(0, x + d) = q(0, y)$.
\end{proof}

\section{Missing Proofs for \cref{sec:weighted-one-counter-automata}}
\label{appendix:sec:weighted-one-counter-automata}
\hillcutting*
\begin{proof}
  Consider a run
  $\rho = (q_0(x_0), \delta_1, q_1(x_1), \ldots, \delta_n, q_n(x_n))$,
  with $q_0(x_0) = p(0)$ and $q_n(x_n) = q(0)$,
  and assume that the length $n$ of its trace $\pi = \delta_1 \cdots \delta_n$
  satisfies $n \geq m^2 |Q|^3$.
  If $x_i < m|Q|^2$ for all $i \in \interval{0}{n}$,
  then some configuration $r(x)$ repeats at least $m+1$ times in the run $\rho$,
  i.e.,
  there exists
  $0 \leq i_0 < \cdots < i_m \leq n$ such that
  $(q_{i_\ell}, x_{i_\ell}) = r(x)$ for all $\ell \in \interval{0}{m}$.
  Taking
  $\alpha = \delta_1 \cdots \delta_{i_0}$,
  $\beta_\ell = \delta_{i_{\ell-1} + 1} \cdots \delta_{i_\ell}$ for all $\ell \in \interval{1}{m}$,
  $\gamma = \theta_m = \cdots = \theta_1 = \varepsilon$ and
  $\eta = \delta_{i_m + 1} \cdots \delta_n$
  concludes the proof of the lemma for this case.

  \smallskip

  Assume now that the run $\rho$ visits a configuration $r(x)$ such that $x \geq m|Q|^2$.
  Intuitively,
  this configuration identifies a ``high hill''.
  By way of the classical hill-cutting technique,
  we get two sequences $i_0, \ldots, i_{m|Q|^2}$ and $j_0, \ldots, j_{m|Q|^2}$
  of positions in $\interval{0}{n}$ verifying
  \begin{itemize}
  \item
    $0 \leq i_0 < \cdots < i_{m|Q|^2} \leq j_{m|Q|^2} < \cdots < j_0 \leq n$, and
  \item
    for every $\ell \in \interval{0}{m|Q|^2}$,
    we have $x_{i_\ell} = x_{j_\ell} = \ell$ and $x_i > \ell$ for all $i \in \interval{i_\ell + 1}{j_\ell - 1}$.
  \end{itemize}
  Consider the pairs of states $(q_{i_\ell}, q_{j_\ell})$ where
  $\ell$ ranges over $\interval{0}{m|Q|^2}$.
  By the pigeonhole principle,
  there exists $0 \leq \ell_0 < \cdots < \ell_m \leq m|Q|^2$ such that
  $(q_{i_{\ell_0}}, q_{j_{\ell_0}}) = \cdots = (q_{i_{\ell_m}}, q_{j_{\ell_m}})$.
  It follows that
  $$
  p(0)
  \xrightarrow{\alpha}
  r(\ell_0)
  \xrightarrow{\beta_1}
  r(\ell_1)
  \cdots
  \xrightarrow{\beta_m}
  r(\ell_m)
  \xrightarrow{\gamma}
  s(\ell_m)
  \xrightarrow{\theta_m}
  \cdots
  s(\ell_1)
  \xrightarrow{\theta_1}
  s(\ell_0)
  \xrightarrow{\eta}
  q(0)
  $$
  where
  $(r, s) = (q_{i_{\ell_0}}, q_{j_{\ell_0}})$,
  $\alpha = \delta_1 \cdots \delta_{i_{\ell_0}}$,
  $\beta_h = \delta_{i_{\ell_{h-1}} + 1} \cdots \delta_{i_{\ell_h}}$ for each $h \in \interval{1}{m}$,
  $\gamma = \delta_{i_{\ell_m} + 1} \cdots \delta_{j_{\ell_m}}$,
  $\theta_h = \delta_{j_{\ell_h + 1}} \cdots \delta_{j_{\ell_{h-1}}}$ for each $h \in \interval{1}{m}$ and
  $\eta = \delta_{j_{\ell_0} + 1} \cdots \delta_n$.
  By construction, we have
  $\pi = \alpha \beta_1 \cdots \beta_m \gamma \theta_m \cdots \theta_1 \eta$ and
  the cycles $\beta_1, \ldots, \beta_m$ and $\theta_1, \ldots, \theta_m$ are all nonempty.

  \smallskip

  It remains to show that the cycles $\beta_h$ and $\theta_h$ can be taken
  any number of times.
  First notice that there is no zero-test in $\beta_1 \cdots \beta_m \gamma \theta_m \cdots \theta_1$,
  since $x_i > \ell_0$ for all $i \in \interval{i_{\ell_0} + 1}{j_{\ell_0} - 1}$.
  Let $h \in \interval{1}{m}$.
  Recall that
  $x_i \geq \ell_{h-1}$ for all $i \in \interval{i_{\ell_{h-1}}}{j_{\ell_{h-1}}}$.
  Hence, the counter remains bounded from below by $\ell_{h-1}$
  in the subruns
  $r(\ell_{h-1}) \xrightarrow{\beta_h} r(\ell_h)$
  and
  $s(\ell_h) \xrightarrow{\theta_h} s(\ell_{h-1})$.
  So we may safely decrease all counter values by $\ell_{h-1}$ in these two subruns
  and obtain two runs
  $r(0) \xrightarrow{\beta_h} r(\ell_h - \ell_{h-1})$
  and
  $s(\ell_h - \ell_{h-1}) \xrightarrow{\theta_h} s(0)$.
  Similarly,
  we may safely decrease all counter values by $\ell_m$ in the subrun
  $r(\ell_m) \xrightarrow{\gamma} s(\ell_m)$
  and obtain a run
  $r(0) \xrightarrow{\gamma} s(0)$.
  We deduce by monotonicity of addition actions that,
  for every $n_1, \ldots, n_m \in \setN$,
  $$
  r(\ell_0)
  \xrightarrow{\beta_1^{n_1}}
  r(\ell_0 + d_1)
  \cdots
  \xrightarrow{\beta_m^{n_m}}
  r(\ell_0 + d_m)
  \xrightarrow{\gamma}
  s(\ell_0 + d_m)
  \xrightarrow{\theta_m^{n_m}}
  \cdots
  s(\ell_0 + d_1)
  \xrightarrow{\theta_1^{n_1}}
  s(\ell_0)
  $$
  where $d_h = n_1 (\ell_1 - \ell_0) + \cdots + n_h (\ell_h - \ell_{h-1}) \geq 0$
  for each $h \in \interval{1}{m}$.
  The observations that
  $p(0) \xrightarrow{\alpha} r(\ell_0)$
  and
  $s(\ell_0) \xrightarrow{\eta} q(0)$
  conclude the proof of the lemma.
\end{proof}

\shortocaruns*
\begin{proof}
  We only briefly sketch the proof of this well-known fact.
  The subcase where $x = y = 0$ immediately follows from \cref{lem:hill-cutting} applied with $m = 1$.
  The general case reduces to this subcase by
  adding $x$ states ``before'' $p$ to increment the counter by $x$ and
  adding $y$ states ``after'' $q$ to decrement the counter by $y$.
\end{proof}

\pumpingshortcycles*
\begin{proof}
  Consider a run
  $\rho = (q_0(x_0), \delta_1, q_1(x_1), \ldots, \delta_n, q_n(x_n))$,
  with $q_0(x_0) = p(0)$ and $q_n(x_n) = q(0)$,
  and assume that the length $n$ of its trace $\pi = \delta_1 \cdots \delta_n$
  satisfies $n \geq 2|Q|^3$.
  We first consider the case where there exists $0 \leq h < k \leq n$
  such that $q_h(x_h) = q_k(x_k)$ and $x_i < 2|Q|^2$ for all $i \in \interval{h}{k}$.
  We may assume, without loss of generality, that
  the configurations $q_i(x_i)$ with $i \in \interval{h}{k-1}$ are pairwise distinct.
  This entails that $k - h \leq |Q| {\cdot} 2|Q|^2 = 2|Q|^3$.
  Taking
  $r = s = q_h$,
  $x = x_h$,
  $d = 0$,
  $\alpha = \delta_1 \cdots \delta_h$,
  $\beta = \delta_{h + 1} \cdots \delta_k$,
  $\gamma = \theta = \varepsilon$ and
  $\eta = \delta_{k + 1} \cdots \delta_n$
  concludes the proof of the lemma for this case.

  \smallskip

  Assume now that $q_h(x_h) \neq q_k(x_k)$ for every $0 \leq h < k \leq n$
  such that $x_i < 2|Q|^2$ for all $i \in \interval{h}{k}$.
  If we had $x_i < 2|Q|^2$ for all $i \in \interval{0}{n}$ then we would get
  that $n < |Q| {\cdot} 2|Q|^2 = 2|Q|^3$, which is impossible.
  So some configuration $q_m(x_m)$ satisfies $x_m \geq 2|Q|^2$.
  Intuitively,
  this configuration identifies a ``high hill''.
  We now apply a tuned hill-cutting technique
  to obtain the desired cycles $\beta$ and $\theta$.
  We split the run $\rho$ by introducing the positions
  $i_0, \ldots, i_{|Q|^2}$ and $j_0, \ldots, j_{|Q|^2}$ defined by
  $$
  \begin{array}{r@{\ \ }c@{\ \ }l}
    i_0 & \eqdef & \max \{i \mid i \leq m \wedge x_i = |Q|^2\}\\
    i_\ell & \eqdef & \min \{i \mid i \geq i_0 \wedge x_i = |Q|^2 + \ell\}
  \end{array}
  \qquad\qquad
  \begin{array}{r@{\ \ }c@{\ \ }l}
    j_0 & \eqdef & \min \{j \mid j \geq m \wedge x_j = |Q|^2\}\\
    j_\ell & \eqdef & \max \{j \mid j \leq j_0 \wedge x_j = |Q|^2 + \ell\}
  \end{array}
  $$
  where $\ell \in \interval{1}{|Q|^2}$.
  It is readily seen that
  $x_{i_\ell} = x_{j_\ell} = |Q|^2 + \ell$ for all $\ell \in \interval{0}{|Q|^2}$
  and that
  $0 \leq i_0 < \cdots < i_{|Q|^2} \leq m \leq j_{|Q|^2} < \cdots < j_0 \leq n$.
  This ad-hoc split is of interest to us because of the two following properties.
  Firstly, $|Q|^2 \leq x_i$ for every $i \in \interval{i_0}{j_0}$.
  Secondly, $|Q|^2 \leq x_i < 2|Q|^2$ for every $i$ such that
  $i_0 \leq i < i_{|Q|^2}$ or $j_{|Q|^2} < i \leq j_0$.
  This second property entails that the configurations
  $q_i(x_i)$ with $i_0 \leq i < i_{|Q|^2}$ are pairwise distinct,
  hence $i_{|Q|^2} - i_0 \leq |Q|^3$,
  and that the configurations
  $q_j(x_j)$ with $j_{|Q|^2} < j \leq j_0$ are pairwise distinct,
  hence $j_0 - j_{|Q|^2} \leq |Q|^3$.
  Consider the pairs of states $(q_{i_\ell}, q_{j_\ell})$ where
  $\ell$ ranges over $\interval{0}{|Q|^2}$.
  By the pigeonhole principle,
  there exists $0 \leq \ell < \ell' \leq |Q|^2$ such that
  $(q_{i_\ell}, q_{j_\ell}) = (q_{i_{\ell'}}, q_{j_{\ell'}})$.
  It follows that
  $$
  p(0)
  \xrightarrow{\alpha}
  r(x)
  \xrightarrow{\beta}
  r(x + d)
  \xrightarrow{\gamma}
  s(x + d)
  \xrightarrow{\theta}
  s(x)
  \xrightarrow{\eta}
  q(0)
  $$
  where
  $r = q_{i_\ell}$,
  $s = q_{j_\ell}$,
  $x = |Q|^2 + \ell$,
  $d = \ell' - \ell$,
  $\alpha = \delta_1 \cdots \delta_{i_\ell}$,
  $\beta = \delta_{i_\ell + 1} \cdots \delta_{i_{\ell'}}$,
  $\gamma = \delta_{i_{\ell'} + 1} \cdots \delta_{j_{\ell'}}$,
  $\theta = \delta_{j_{\ell'} + 1} \cdots \delta_{j_\ell}$ and
  $\eta = \delta_{j_\ell + 1} \cdots \delta_n$.
  By construction, we have
  $\pi = \alpha \beta \gamma \theta \eta$
  and
  $x+d = |Q|^2 + \ell' \leq 2|Q|^2$.
  We also have
  $|\beta \theta| = i_{\ell'} - i_\ell + j_{\ell} - j_{\ell'}$,
  hence,
  $\beta \theta \neq \varepsilon$ and
  $|\beta \theta| \leq i_{|Q|^2} - i_0 + j_0 - j_{|Q|^2} \leq 2|Q|^3$.

  \smallskip

  It remains to show that the cycles $\beta$ and $\theta$ can be taken
  any number of times.
  First notice that there is no zero-test in $\beta \gamma \theta$,
  since $x_i \geq |Q|^2 \geq 1$ for every $i \in \interval{i_0}{j_0}$.
  Together with the observation that $d \geq 0$,
  this entails by monotonicity of addition actions that
  $
  r(x)
  \xrightarrow{\beta^n}
  r(x + n d)
  \xrightarrow{\gamma}
  s(x + n d)
  \xrightarrow{\theta^n}
  s(x)
  $
  for every $n \geq 1$.
  For the remaining case of $n = 0$,
  recall that the counter remains bounded from below by $|Q|^2$
  in the subrun
  $r(x + d) \xrightarrow{\gamma} s(x + d)$.
  Since $d = \ell' - \ell \leq |Q|^2$,
  we may safely decrease all counter values by $d$ in this subrun
  and obtain a run $r(x) \xrightarrow{\gamma} s(x)$.
  This concludes the proof of the lemma.
\end{proof}

\section{Missing Proofs for \cref{sec:linear-path-scheme-to-equations}}
\label{appendix:sec:linear-path-scheme-to-equations}
\begin{claim}
  \label{claim:pathVASS}
  Let $\pi$ be a path containing no zero-test transition, from a state $p$ to a state $q$. For every $\vec{x},\vec{y}\in\setN^d$ we have:
  $$p(\vec{x})\xrightarrow{\pi}q(\vec{y})~~\Longleftrightarrow~~\vec{x}\geq \vec{m}_\pi\wedge \vec{y}=\vec{x}+\disp{\pi}$$
\end{claim}
\begin{claimproof}
  This is a classical result. Observe that $p(\vec{x})\xrightarrow{\pi}q(\vec{y})$ if, and only if, $\vec{y}=\vec{x}+\disp{\pi}$ and for every prefix $\pi$ of $\alpha$, and for every $i\in\interval{1}{d}$, we have $\vec{x}(i)+\disp{\pi}(i)\geq 0$.
\end{claimproof}

\begin{claim}
  \label{claim:pathT}
  Let $\pi$ be a feasible path containing a zero-test transition, from a state $p$ to a state $q$. For every $\vec{x},\vec{y}\in\setN^d$ we have: $$p(\vec{x})\xrightarrow{\pi}q(\vec{y})~~\Longleftrightarrow~~\vec{x}\geq_1\vec{m}_\pi\wedge  \vec{y}=\vec{x}+\disp{\pi}$$
\end{claim}
\begin{claimproof}
  Let us consider a feasible path $\pi$ containing a zero-test transition, from a state $p$ to a state $q$. The path $\pi$ can be decomposed into $\pi_0\delta_1\pi_1\cdots \delta_k\pi_k$ where $k\geq 1$, $\delta_1,\ldots,\delta_k$ are zero-test transitions and $\pi_0,\ldots,\pi_k$ are paths containing no zero-test transition.
  
  Assume first that $p(\vec{x})\xrightarrow{\pi}q(\vec{y})$. Notice that $\vec{y}=\vec{x}+\disp{\pi}$ and $\vec{x}(i)+\disp{\alpha}(i)\geq 0$ for every prefix $\alpha$ of $\pi$ and for every $i\in\interval{1}{d}$. It follows that $\vec{x}\geq \vec{m}_\pi$. Since $\pi_0\delta_1$ is a prefix of $\pi$, we derive that $\vec{x}(1)+\disp{\pi_0}(1)=0$. Moreover, as $\pi_0$ is a prefix of $\pi$, the definition of $\vec{m}_\pi$ shows that $\vec{m}_\pi\geq -\disp{\pi_0}$. We deduce that $\vec{x}(1)=-\disp{\pi_0}(1)=\vec{m}_\pi(1)$. We have proved that $\vec{x}\geq_1\vec{m}_\pi$.
  
  Conversely, assume that $\vec{x}\geq_1\vec{m}_\pi\wedge  \vec{y}=\vec{x}+\disp{\pi}$. Since $\pi$ is feasible, we deduce from the previous paragraph that $\vec{m}_\pi(1)=-\disp{\pi_0}(1)$. By projecting away the path $\pi$ on counters $i\in \interval{2}{d}$, the proof of $p(\vec{x})\xrightarrow{\pi}q(\vec{y})$ reduces from \cref{claim:pathVASS} to the special case $d=1$ which is immediate since $\vec{x}(1)=-\disp{\pi_0}(1)$ is the unique possible initial value for the counter that makes the path $\pi$ feasible.
\end{claimproof}

\pathinTVASS*
\begin{proof}
  The lemma immediately follows from \cref{claim:pathVASS,claim:pathT}.
\end{proof}

\begin{claim}
  \label{claim:cycleVASS}
  Let $\beta$ be a cycle containing no zero-test transition, on a state $q$. For every $\vec{x},\vec{y}\in\setN^d$ and $n\in\setN\setminus\{0\}$, we have:
  $$q(\vec{x})\xrightarrow{\beta^n}q(\vec{y}) ~~\Longleftrightarrow~~\vec{x}\geq\vec{m}_\beta\wedge \vec{y}\geq \vec{m}_{\overline{\beta}}\wedge \vec{y}=\vec{x}+n\disp{\beta}$$
\end{claim}
\begin{claimproof}
  Let $\vec{x},\vec{y}\in\setN^d$ and $n\in\setN\setminus\{0\}$. Assume first that $q(\vec{x})\xrightarrow{\beta^n}q(\vec{y})$. It follows that $\vec{y}=\vec{x}+\disp{\beta^n}=\vec{x}+n\disp{\beta}$. Moreover, since $n\geq 1$, there exists $\vec{x}',\vec{y}'\in\setN^d$ such that $q(\vec{x})\xrightarrow{\beta}q(\vec{x}')$ and  $q(\vec{y})\xrightarrow{\overline{\beta}}q(\vec{y}')$. \cref{claim:pathVASS} shows that $\vec{x}\geq \vec{m}_{\beta}$ and $\vec{y}\geq \vec{m}_{\overline{\beta}}$. Conversely, assume that
$\vec{x}\geq\vec{m}_\beta$, $\vec{y}\geq \vec{m}_{\overline{\beta}}$, and $\vec{y}=\vec{x}+n\disp{\beta}$. Let us first prove that $\vec{x}(i)+r\disp{\beta}(i)\geq \vec{m}_{\beta}(i)$ for every $r\in\interval{0}{n-1}$ and for every $i\in\interval{1}{d}$. Observe that if $\disp{\beta}(i)\geq 0$ the inequality is immediate. So, let us assume that $\disp{\beta}(i)<0$. As $\vec{y}=\vec{x}+n\disp{\beta}$, we deduce that $\vec{x}+(r+1)\disp{\beta}=\vec{y}-(n-r-1)\disp{\beta}$. From $\disp{\beta}(i)<0$ and $\vec{y}(i)\geq \vec{m}_{\overline{\beta}}$ we get $\vec{x}(i)+(r+1)\disp{\beta}(i)\geq \vec{m}_{\overline{\beta}}(i)$. From \cref{lem:alphaalpha}, we derive $\vec{x}(i)+r\disp{\beta}(i)\geq \vec{m}_\beta(i)$. We have proved that $\vec{x}+r\disp{\beta}\geq \vec{m}_\beta$ for every $r\in\interval{0}{n-1}$. \cref{claim:pathVASS} shows that $q(\vec{x}+r\disp{\beta})\xrightarrow{\beta}q(\vec{x}+(r+1)\disp{\beta})$ for every $r\in\interval{0}{n-1}$. We have proved that $q(x)\xrightarrow{\beta^n}q(\vec{y})$.
\end{claimproof}

\begin{claim}
  \label{claim:cycleT}
  Let $\beta$ be a feasible cycle containing a zero-test transition, on a state $q$. For every $\vec{x},\vec{y}\in\setN^d$ and $n\in\setN\setminus\{0\}$, we have:
  $$q(\vec{x})\xrightarrow{\beta^n}q(\vec{y}) ~~\Longleftrightarrow~~\vec{x}\geq_1\vec{m}_\beta\wedge \vec{y}\geq_1 \vec{m}_{\overline{\beta}}\wedge \vec{y}=\vec{x}+n\disp{\beta}$$
\end{claim}
\begin{claimproof}
  By considering the $(d-1)$-VASS obtained by removing the first counter, the proof of this claim reduces to the special case $d=1$ thanks to \cref{claim:cycleVASS}. Let $x,y\in\setN$ and $n\in\setN\setminus\{0\}$. Assume first that $q(x)\xrightarrow{\beta^n}q(y)$. It follows that $y=x+n\disp{\beta}$. Moreover, since $n\geq 1$, there exists $x',y'\in\setN$ such that $q(x)\xrightarrow{\beta}q(x')$ and  $q(y)\xrightarrow{\overline{\beta}}q(y')$. \cref{claim:pathT} shows that $x= \vec{m}_{\beta}$ and $y= \vec{m}_{\overline{\beta}}$. Conversely, assume that $x=\vec{m}_\beta$, $y=\vec{m}_{\overline{\beta}}$, and $y=x+n\disp{\beta}$. As $x=\vec{m}_\beta$, $\vec{m}_{\overline{\beta}}=\vec{m}_\beta+\disp{\beta}$ from \cref{lem:alphaalpha}, we deduce that $(n-1)\disp{\beta}=0$. If $n=1$, \cref{claim:pathT} shows that $q(x)\xrightarrow{\beta}q(y)$, and if $n\geq 2$ from $(n-1)\disp{\beta}=0$ we get $\disp{\beta}=0$. In particular \cref{claim:pathT} shows that $q(x)\xrightarrow{\beta^n}q(x)=q(y)$.
\end{claimproof}

\cycleinTVASS*
\begin{proof}
  The lemma immediately follows from \cref{claim:cycleVASS,claim:cycleT}.
\end{proof}

\section{Complexity of Boundedness and Termination}
\label{appendix:sec:complexity-results}

We first focus on the boundedness problem. Let us recall that a $2$-TVASS is bounded from an initial configuration $p(\vec{x})$ if the set of configurations $q(\vec{y})$ reachable from $p(\vec{x})$ is finite. We show that the boundedness problem is decidable in polynomial space by first proving that the reachable configurations of bounded $2$-TVASS have a polynomial size.

A linear path scheme $L$ is said to be \emph{bounded} from a configuration $p(\vec{x})$ if the set of configurations $q(\vec{y})$ satisfying $p(\vec{x})\xrightarrow{\pi}q(\vec{y})$ where $\pi$ is a prefix of a path in $L$ is finite.
 \begin{lemma}\label{lem:boundedcycle}
   Assume that a linear path scheme $\beta^*$ is bounded from a configuration $q(\vec{x})$. Every configuration $q(\vec{y})$ such that $q(\vec{x})\xrightarrow{\beta^n}q(\vec{y})$ with $n\geq 2$ satisfies $\norm{\vec{y}}\leq (1+\norm{\disp{\beta}})\norm{\vec{x}}$.
 \end{lemma}
 \begin{proof}
   Assume that $q(\vec{x})\xrightarrow{\beta^n}q(\vec{y})$ with $n\geq 2$. It follows that $\beta$ is feasible. Moreover, since $n\geq 2$, \cref{lem:cycle-in-TVASS} shows that if $\beta$ contains a zero-test action, then $\disp{\beta}(1)=0$. If $\disp{\beta}\geq \vec{0}$, \cref{lem:cycle-in-TVASS} shows that $q(\vec{x})\xrightarrow{\beta^m}q(\vec{x}+m\disp{\beta})$ for every $m\in\setN$. Since $\beta^*$ is bounded from $q(\vec{x})$, the set $\{\vec{x}+m\disp{\beta} \mid m\in\setN\}$ is finite. We deduce that $\disp{\beta}=\vec{0}$. In that case, $\vec{y}=\vec{x}$ and we are done. So, we can assume that $\disp{\beta}\not\geq \vec{0}$. So there exists $i\in\{1,\ldots,d\}$ such that $\disp{\beta}(i)\leq -1$. From $\vec{x}+n\disp{\beta}=\vec{y}\geq\vec{0}$, we derive $\vec{x}(i)\geq n$. Hence $n\leq \norm{\vec{x}}$. It follows that $\norm{\vec{y}}\leq \norm{\vec{x}}(1+\norm{\disp{\beta}})$ and we are done. 
 \end{proof}

 \begin{lemma}\label{lem:boundedlps}
   Let $L$ be a linear path scheme bounded from a configuration $p(\vec{x})$. For every configuration $q(\vec{y})$ such that $p(\vec{x})\xrightarrow{L}q(\vec{y})$, we have:
   $$\norm{\vec{y}}\leq (\norm{\vec{x}}+|L|\norm{\Sigma})(1+|L|\norm{\Sigma})^{|L|_*}$$
   \end{lemma}
   \begin{proof}
     We prove the lemma by induction on a natural number $k$ bounding $|L|_*$. The rank $0$ is trivial since $|L|_*=0$ implies that $L$ is a single path. Let us assume the rank $k-1$ proved for some $k\geq 1$ and let $L=\alpha_0\beta_1^*\alpha_1\ldots \beta_k^*\alpha_k$ be a linear path scheme bounded from a configuration $p(\vec{x})$. Assume that $p(\vec{x})\xrightarrow{\pi}q(\vec{y})$ for some path $\pi$ in $L$. There exists $n_1,\ldots,n_k\in\setN$ such that:
     $$\pi=\alpha_0\beta_1^{n_1}\alpha_1\ldots \beta_k^{n_k}\alpha_k$$
     Notice that if for some $j\in\{1,\ldots,k\}$ we have $n_j\in\{0,1\}$, then the linear path scheme $L'$ obtained from $L$ be replacing $\alpha_{j-1}\beta_j^*\alpha_j$ by $\alpha_{j-1}\alpha_j$ or $\alpha_{j-1}\beta_j\alpha_j$ is a linear path scheme bounded from $p(\vec{x})$ such that $p(\vec{x})\xrightarrow{L'}q(\vec{y})$ and we are done by induction since $|L'|_*=k-1$ and $|L'|\leq|L|$.
     So, we can assume that $n_1,\ldots,n_k\geq 2$. Let us introduce the linear path scheme $L'=\alpha_0\beta_1^*\alpha_1\ldots \beta_{k-1}^*\alpha_{k-1}$ and observe that there exist configurations $r(\vec{a})$ and $r(\vec{b})$ such that:
     $$p(\vec{x})\xrightarrow{L'}r(\vec{a})\xrightarrow{\beta_k^{n_k}}r(\vec{b})\xrightarrow{\alpha_k}q(\vec{y})$$
     Since $L'$ is bounded from $p(\vec{x})$, we deduce by induction that
     $$\norm{\vec{a}}\leq (\norm{\vec{x}}+|L'|\norm{\Sigma})(1+|L'|\norm{\Sigma})^{k-1}$$
     Since $L$ is bounded from $p(\vec{x})$, observe that $\beta_k^*$ is bounded from $r(\vec{a})$. \cref{lem:boundedcycle} shows that $\norm{\vec{b}}\leq (1+|L|\norm{\Sigma})\norm{\vec{a}}$. Moreover, from $\vec{y}=\vec{b}+\disp{\alpha_k}$, we get $\norm{\vec{y}}\leq \norm{\vec{b}}+|\alpha_k|\norm{\Sigma}$. We have proved the induction.
  \end{proof}

  We deduce a polynomial bound on the size of reachable configurations.
  \begin{corollary}\label{cor:boundedTVASS}
    For every $2$-TVASS $\vass$ bounded from an initial configuration $p(\vec{x})$, the reachable configurations $q(\vec{y})$ from $p(\vec{x})$ satisfy:
    $$\norm{\vec{y}}\leq (1+\norm{\vec{x}})(|Q|+\norm{\Sigma})^{O(|Q|^3)}$$
  \end{corollary}
  \begin{proof}
    Let $c\geq 1$ be a constant satisfying \cref{cor:lps}, i.e., such that $O(1)\leq c$ and $O(|Q|^3)\leq c|Q|^3$.
    Let us consider a $2$-TVASS $\vass=(Q,\Sigma,\Delta)$, and let $N=|Q|+\norm{\Sigma}$. Let us consider a configuration $q(\vec{y})$ such that $p(\vec{x})\xrightarrow{*}q(\vec{y})$.
    The case where $\Sigma \subseteq \{\vec{0}, \ztest\}$ is trivial (as $\vec{x} = \vec{y}$ in that case),
    so we assume that $\norm{\Sigma} \geq 1$, hence $N \geq 2$, for the remainder of the proof.
    \cref{cor:lps} shows that there exists a linear path scheme $L$ with $|L|\leq N^c$ and $|L|_*\leq c|Q|^3$ such that $p(\vec{x})\xrightarrow{L}q(\vec{y})$. Since $\vass$ is bounded from $p(\vec{x})$ it follows that $L$ is bounded from $p(\vec{x})$. From \cref{lem:boundedlps}, we derive that $\norm{\vec{y}}$ is bounded by:
    $$
    (\norm{\vec{x}}+|L|\norm{\Sigma})(1+|L|\norm{\Sigma})^{|L|_*}
    \leq
    (\norm{\vec{x}}+N^{c+1})(1+N^{c+1})^{c|Q|^3}
    \leq
    (1+\norm{\vec{x}}) N^{5c^2|Q|^3}
    $$
    This concludes the proof of the lemma.
  \end{proof}

  It follows that a $2$-TVASS is not bounded from an initial configuration $p(\vec{x})$ if there exists a reachable configuration $q(\vec{y})$ that exceeds the bound $B$ introduced in \cref{cor:boundedTVASS}. In that case, there exists a run such that all intermediate configurations are bounded by $B$ except the last one that exceeds $B$. Note that this last configuration is bounded by $B+\norm{\Sigma}$. We deduce that a bounded polynomial-space exploration of the reachability set of a $2$-TVASS provides a way to decide the boundedness problem. We have proved the following theorem.
 \begin{theorem}\label{thm:boundedness}
   The boundedness problem for $2$-TVASS is PSPACE-complete.
 \end{theorem}

 \medskip
 
 Let us now focus on the termination problem. Recall that a $2$-TVASS is not terminating from an initial configuration $p(\vec{x})$ if there exists an infinite run from $p(\vec{x})$, i.e., an infinite sequence $(q_n(\vec{x}_n)_{n\in\setN}$ of configurations such that $q_0(\vec{x}_0)=p(\vec{x})$ and $q_n(\vec{x}_n) \rightarrow q_{n+1}(\vec{x}_{n+1})$ for all $n \in \setN$. Since a $2$-TVASS is finitely branching, we observe that if a $2$-TVASS is not bounded from an initial configuration $p(\vec{x})$ then it is not terminating from $p(\vec{x})$. It follows that the termination problem reduces to the special bounded case. In that case, \cref{cor:boundedTVASS} shows that the size of reachable configurations are polynomially bounded. It follows that non-termination is equivalent to the existence of a reachable configuration $q(\vec{y})$ from $p(\vec{x})$ such that $q(\vec{y}) \rightarrow \rcomp \xrightarrow{*} q(\vec{y})$. Since such a property can be decided in polynomial space, we deduce the following result.
 \begin{theorem}\label{thm:termination}
   The termination problem for $2$-TVASS is PSPACE-complete.
 \end{theorem}

\end{document}